\documentclass[journal,twoside,web]{ieeecolor}
\usepackage{generic}
\usepackage{cite}
\usepackage{textcomp}

\usepackage{amssymb,latexsym,amsfonts,amsmath,bbm}
\usepackage{mathrsfs}
\usepackage{graphicx}
\usepackage{paralist}
\usepackage{dsfont}
\usepackage{eso-pic}
\usepackage{epsfig}
\usepackage{mathtools}
\usepackage{epstopdf}
\usepackage{lipsum}
\usepackage{cite}
\usepackage{subfig}%%

\usepackage{optidef}

\usepackage{tikz}
\usepackage{pgfplots}
\definecolor{denim}{rgb}{0.08, 0.38, 0.74}
\definecolor{deeppink}{rgb}{1.0, 0.08, 0.58}
\usetikzlibrary{positioning}
\usetikzlibrary{
	arrows.meta,                        % for arrow tips
	backgrounds,                        % for background layer
	ext.paths.ortho,                    % for ortho paths
	ext.positioning-plus,               % for 
	ext.node-families.shapes.geometric, % loads ext.node-families and
	calc}                               % for ($<calculations>$)
\tikzset{
	basic box/.style={
		shape=rectangle, rounded corners, align=center, draw=#1, fill=#1!25, minimum height=20mm, minimum width = 20mm},
	header node/.style={
		node family/width=header nodes,
		font=\strut\Large\ttfamily,
		text depth=+.3ex, fill=white, draw},
	header/.style={%
		inner ysep=+1.5em,
		append after command={
			\pgfextra{\let\TikZlastnode\tikzlastnode}
			node [header node] (header-\TikZlastnode) at (\TikZlastnode.north) {#1}
			node [span=(\TikZlastnode)(header-\TikZlastnode)]
			at (fit bounding box) (h-\TikZlastnode) {}
		}
	},
	fat blue line/.style={ultra thick, blue}
}
\usetikzlibrary{matrix}

\usepackage{algorithm}
\usepackage{algorithmic}

\newcounter{theorem}
\newcounter{definition}
\newcounter{lemma}
\newcounter{claim}
\newcounter{problem}
\newcounter{proposition}
\newcounter{corollary}
\newcounter{construction}
\newcounter{example}
\newcounter{xca}
\newcounter{comments}
\newcounter{remark}
\newcounter{assumption}

\newtheorem{theorem}[theorem]{Theorem}
\newtheorem{lemma}[lemma]{Lemma}

\newtheorem{problem}[problem]{Problem}

\newtheorem{definition}[definition]{Definition}

\newtheorem{remark}[remark]{Remark}
\newtheorem{assumption}[assumption]{Assumption}

\numberwithin{equation}{section}

\DeclareFontFamily{U}{stix2bb}{}
\DeclareFontShape{U}{stix2bb}{m}{n} {<-> stix2-mathbb}{}

\usepackage[many]{tcolorbox}
\usetikzlibrary{calc}
\tcbuselibrary{skins}

\newtcolorbox{resp}[1][]{%
	enhanced jigsaw,%
	colback=gray!5!white,%
	colframe=gray!80!black,%
	size=small,%
	boxrule=1pt,%
	halign title=flush center,%
	coltitle=black,%
	breakable,%
	drop shadow=black!50!white,%
	attach boxed title to top left={xshift=1cm,yshift=-\tcboxedtitleheight/2,yshifttext=-\tcboxedtitleheight/2},%
	minipage boxed title=3cm,%
	boxed title style={%
		colback=white,%
		size=fbox,%
		boxrule=1pt,%
		boxsep=2pt,%
		underlay={%
			\coordinate (dotA) at ($(interior.west) + (-0.5pt,0)$);
			\coordinate (dotB) at ($(interior.east) + (0.5pt,0)$);
			\begin{scope}[gray!80!black]
				\fill (dotA) circle (2pt);
				\fill (dotB) circle (2pt);
			\end{scope}
		}%
	},%
	#1%
}

\usepackage{xspace}

\def\BS#1{{\textcolor{black}{#1}}}

\DeclareRobustCommand{\legendsquare}[1]{%
	\textcolor{#1}{\rule{1.5ex}{1.5ex}}%
}%%

\definecolor{unsafe}{rgb}{0.9, 0, 0.1}%%
\definecolor{initt}{rgb}{0.1, 0.3, 1}%%
\definecolor{state}{rgb}{0, 0.9, .7}%%

\def\BibTeX{{\rm B\kern-.05em{\sc i\kern-.025em b}\kern-.08em
		T\kern-.1667em\lower.7ex\hbox{E}\kern-.125emX}}
	
\usepackage{etoolbox}

\makeatletter
\AtBeginDocument{%
	\pagestyle{headings}
	
	\patchcmd{\@oddhead}{\\[-19pt]}{\\[-8pt]}{}{}%
	\patchcmd{\@evenhead}{\\[-19pt]}{\\[-8pt]}{}{}%
}
\makeatother

\definecolor{blue(ryb)}{rgb}{0.01, 0.28, 1.0}
\definecolor{fashionfuchsia}{rgb}{0.96, 0.0, 0.63}
\makeatletter
\let\NAT@parse\undefined
\makeatother
\usepackage[colorlinks=true, citecolor=red, linkcolor=blue, final]{hyperref}
\renewcommand{\theequation}{\arabic{equation}}
\counterwithout*{equation}{section}

\makeatletter
\def\@opargbegintheorem#1#2#3{\textit{#1\ #2} \textit{(#3):}}
\makeatother

\begin{document}
	
\title{Data-Driven Control of Large-Scale Networks with Formal Guarantees: A Small-Gain Free Approach}
\author{Behrad Samari, \IEEEmembership{Student Member,~IEEE}, Amy Nejati, \IEEEmembership{Senior Member,~IEEE}, and 	Abolfazl Lavaei, \IEEEmembership{Senior Member,~IEEE}
	\thanks{All the authors are with the School of Computing, Newcastle University, United Kingdom. Email:
		\texttt{b.samari2@newcastle.ac.uk},
		\texttt{amy.nejati@newcastle.ac.uk},
		\texttt{abolfazl.lavaei@newcastle.ac.uk}. }
}
\maketitle
\begin{abstract}
This paper offers a \emph{data-driven divide-and-conquer} strategy to analyze large-scale interconnected networks, characterized by both \emph{unknown} mathematical models and interconnection topologies. Our data-driven scheme treats an unknown network as an interconnection of individual agents (a.k.a. subsystems) and aims at constructing their \emph{symbolic models}, referred to as \emph{discrete-domain} representations of unknown agents, by collecting data from their trajectories. The primary objective is to synthesize a control strategy that guarantees desired behaviors over an unknown network by employing local controllers, derived from symbolic models of individual agents. To achieve this, we leverage the concept of \emph{alternating sub-bisimulation function (ASBF)} to capture the closeness between the state trajectories of each unknown agent and its data-driven symbolic model. Under a newly developed \emph{data-driven compositional} condition, we then establish an \emph{alternating bisimulation function (ABF)} between an unknown network and its symbolic model based on ASBFs of individual agents while providing \emph{correctness guarantees}. Despite the sample complexity in existing work being \emph{exponential} with respect to the network size,  we demonstrate that our divide-and-conquer strategy significantly reduces it to a \emph{linear scale} with respect to the number of agents. We also showcase that our data-driven compositional condition does not necessitate the traditional small-gain condition, which demands precise knowledge of the interconnection topology for its fulfillment. We apply our data-driven findings to \BS{three} benchmarks comprising unknown networks with an \emph{arbitrary, a priori undefined} number of agents and unknown interconnection topologies. 
\end{abstract}


\section{Introduction} \label{Introduction}
\IEEEPARstart{I}{nterconnected} networks have emerged as invaluable assets for modeling a diverse array of practical engineering systems in recent decades. These networks serve crucial roles in various domains, ranging from automated vehicles to biological processes and energy infrastructures. In these applications, the number of involved agents can be exceedingly large, potentially \emph{unknown}, or subject to variation over time as agents might join or leave the network. Without careful consideration and rigorous mitigation strategies, scalability challenges of such complex networks have the potential to significantly degrade the system's performance~\cite{bamieh2002distributed,jovanovic2005ill}.

To offer formal assurances over the behavior of complex dynamical systems, \emph{symbolic models} (a.k.a. finite abstractions) have been introduced as abstract descriptions of original dynamics in discrete domains~\cite{tabuada2009verification,zamani2011symbolic,coogan2016finite,majumdar2020abstraction}. These techniques are instrumental in designing controllers aimed at enforcing complex specifications (\emph{e.g.,} a vehicle reaches its destination while avoiding obstacles) that are challenging to address using traditional control design methods. More precisely, symbolic models can be employed as suitable substitutes for original systems to provide formal analyses over the discrete domain. The obtained results can then be translated back to the original realm within a guaranteed error bound quantified between the state trajectories of the two systems using the notion of \emph{alternating bisimulation functions}~\cite{tabuada2009verification}. This approach ensures that the original system satisfies the same specifications as its symbolic model within a quantified error threshold.

Symbolic models are typically classified into two types: \emph{sound and complete} abstractions \cite{tabuada2009verification}. Complete abstractions provide both \emph{sufficient and necessary} guarantees, meaning that a controller exists to enforce a desired property on the abstraction \emph{if and only if} a controller exists for the same specification on the original system. In contrast, sound abstractions offer only \emph{sufficient} guarantees; thus, the inability to synthesize a controller through a sound abstraction does not necessarily indicate that a controller is absent for the original system.

While symbolic model techniques prove effective in analyzing dynamical systems, constructing them for large-scale networks introduces significant challenges that can be considered twofold. Firstly, the computational complexity of constructing symbolic models increases \emph{exponentially} with the system's dimension, often leading to the \emph{curse of dimensionality}. Additionally, and of greater significance, knowing the precise mathematical model of the underlying system is a prerequisite to constructing such abstract models. \emph{Compositional} techniques have subsequently been introduced to address the first challenge by constructing the symbolic model of an interconnected network based on symbolic models of its individual agents~\cite{swikir2019compositional,lavaei2019automated,nejati2023formal}. \BS{System identification, which has a rich literature, offers a means to tackle the second problem. Such methods, referred to as \emph{indirect} data-driven approaches~\cite{hou2013model, dorfler2022bridging}, typically involve first identifying a valid model  and then constructing finite abstractions together with their associated functions and relations. However, identifying a large-scale network as a whole is computationally prohibitive due to its high dimensionality; hence, subsystem-level identification is often required. Furthermore, while many well-established techniques exist for identifying models for linear subsystems or specific classes of nonlinear subsystems, current nonlinear system identification techniques still exhibit several limitations~\cite{kerschen2006past}. Even when subsystem models are successfully identified, one must still construct ASBFs, resulting in a \emph{two-stage computational burden}. Moreover, extending the analysis to the network level requires explicit knowledge of the interconnection topology, whose identification, despite a rich literature~\cite{materassi2010topological, nabi2012network, shahrampour2014topology}, remains an equally challenging task.}

\textbf{Innovative findings.} Inspired by these two key challenges, the primary contribution of this work lies in developing a \emph{direct} data-driven approach within a \emph{compositional framework} to circumvent the system identification step and directly utilize data for constructing symbolic models and similarity relations for large-scale networks, characterized by unknown models as well as unknown interconnection topologies. Our framework's backbone hinges on alternating bisimulation functions that measure the closeness between the trajectories of an interconnected network and its symbolic model. Within our data-driven scheme, we operate at the subsystem level and initially reframe the conditions of alternating sub-bisimulation functions as a robust optimization program (ROP). Considering the emergence of unknown subsystem dynamics within the ROP, which makes it intractable, we collect a set of data pairs from the trajectories of unknown subsystems and introduce a scenario optimization program (SOP) tailored to each ROP. We then introduce an innovative \emph{data-driven compositional} technique to construct an alternating bisimulation function for a network with an unknown \emph{interconnection  topology} via alternating sub-bisimulation functions of smaller subsystems, derived from data, while offering \emph{correctness guarantees}. Our data-driven technique notably reduces the sample complexity in existing studies, shifting from an \emph{exponential} dependence on the network size to a \emph{linear scale} with respect to the number of agents. Our compositional condition, derived from data, also eliminates the need for the traditional small-gain condition, which requires exact knowledge of the interconnection topology (cf.~\cite[Eq.~(3.7)]{swikir2019compositional}). Our proposed framework enables the construction of symbolic models for interconnected networks with an \emph{arbitrary, a priori undefined} number of subsystems, as demonstrated in the case study section.

\textbf{Related studies.} Several efforts have addressed either \emph{stability analysis} of interconnected networks~\cite{chaillet2014strong,dashkovskiy2007iss,tanner2002stability,mironchenko2020input} or \emph{symbolic model} constructions of interconnected networks~\cite{swikir2019compositional,meyer2017compositional}, relying on traditional small-gain conditions. However, these approaches assume knowledge of the precise mathematical model of the network, which is typically unknown in practice. Furthermore, to fulfill the traditional small-gain compositionality condition, these methods necessitate precise knowledge of the \emph{interconnection topology}, a challenge in real-world applications (\emph{e.g.,} in a network of vehicles where each vehicle can join or leave the network over time).

\BS{Dissipativity-based compositional reasoning provides another well-established means for the stability analysis and symbolic model construction of interconnected networks~\cite{lavaei2022dissipativity}. However, it requires explicit knowledge of the interconnection topology to infer network-level conclusions from subsystem-level results~\cite{2016Murat}. While some topology-free variants have been explored in~\cite{willems1972dissipative, haddad2008nonlinear}, these approaches rely on the assumption that the interconnection is neutral (lossless). While valuable, this neutrality assumption can be restrictive in practice, particularly in data-driven settings, where verifying such a property from finite data is challenging and requires further investigation. In contrast, our proposed data-driven framework remains topology-free and does not rely on any interconnection assumptions; instead, the effect of the interconnection structure is inherently captured by the data used to solve the problem (cf.~Remark~\ref{rem:Lip_why} and SOP~\eqref{SOP}).}

There has been limited work on constructing symbolic models using data. Existing results include symbolic model construction via a Gaussian process approach~\cite{hashimoto2022learning}, data-driven symbolic model of monotone systems with disturbances~\cite{makdesi2021efficient}, data-driven construction of symbolic models for verification and control of unknown systems~\cite{devonport2021symbolic, coppola2022data, coppola2024data, banse2024data}, and computation of growth bounds from data for constructing symbolic models~\cite{kazemi2024data,ajeleye2023data}. Since all these data-driven efforts aim at constructing symbolic models for \textit{monolithic systems}, they are not applicable to large-scale networks due to the \emph{exponential} curse of sample complexity, typically limiting abstractions to systems with dimensions up to 5.

Recent efforts have aimed to construct symbolic models for interconnected networks using data. Specifically, \cite{ajeleye2024data} introduces a data-driven method for finite abstractions of \emph{relatively} high-dimensional systems. However, the results proposed in \cite{ajeleye2024data} require approximating the interconnection map under the assumption that its exact Lipschitz constant is estimated with confidence~$1$ in the limit, a restrictive condition in practical applications that also adds computational complexity for large interconnection maps. More importantly, although the abstraction is constructed compositionally, the control synthesis is performed monolithically, leading to significant scalability issues for high-dimensional systems. This limitation is evident in the second case study in \cite{ajeleye2024data}, where only an $8$-dimensional system is considered. Our proposed framework, however, does not encounter any of these difficulties as it neither requires an approximation of the interconnection map nor relies on monolithic control synthesis (cf. our case studies with over 10000 subsystems). Additionally, while \cite{ajeleye2024data} provides sound abstractions (\emph{sufficient} guarantees), our method yields complete abstractions, offering both \emph{sufficient and necessary} guarantees. The work \cite{lavaei2023symbolic} also proposes a data-driven approach for constructing symbolic models for interconnected networks. However, their results come with a probabilistic confidence, whereas our approach provides deterministic correctness guarantees. Furthermore, the method in  \cite{lavaei2023symbolic} requires the network topology to be known, an a posteriori check of the small-gain condition (see circularity condition in \cite[Eq. (19)]{lavaei2023symbolic}), and a priori knowledge of the number of subsystems, all of which are relaxed in our work.

\section{Problem Description}\label{Problem_Description}

\subsection{Notation}
Sets of real, positive, and non-negative real numbers are denoted by $\mathbb{R}$, $\mathbb{R}^+$, and $\mathbb{R}_0^+$, respectively. Non-negative and positive integers are, respectively, represented by $\mathbb{N} := \{0,1,2,\ldots\}$ and $\mathbb{N}^+=\{1,2,\ldots\}$. Given $N$ vectors $x_i \in \mathbb{R}^{n_i}$, the corresponding column vector of dimension $\sum_i n_i$ is signified by $x=[x_1;\ldots;x_N]$.  The Euclidean norm of a vector $x\in\mathbb{R}^{n}$ is denoted by $\Vert x\Vert$. The Cartesian product of sets $X_i, i\in\{1,\ldots,N\}$, is represented by $\prod_{i=1}^{N}X_i$. Given sets $X$ and $Y$, a relation $\mathscr{R} \subseteq X \times Y$ is a subset of the Cartesian product $X \times Y$ that relates $x \in X$ to $y \in Y$ if $(x, y) \in \mathscr{R}$, equivalently denoted by $x \mathscr{R} y$.

\subsection{Interconnected Networks}
We initiate by defining subsystems (a.k.a. agents) as the fundamental building blocks, which are thereafter interconnected to form a large-scale network.
\begin{definition} \label{discrete}
	A discrete-time control subsystem (dt-CS) $\Upsilon\!_{i}, i \in \{1,\dots,M\}$,  can be represented as
	\begin{equation*}\label{eq:dt-SCS}
		\Upsilon\!_{i}=(X_i,U_i,W_i,f_i),
	\end{equation*}
	where:
	\begin{itemize}
		\item $X_i\subseteq \mathbb R^{n_i}$ is the state set of dt-CS $\Upsilon\!_{i}$; 
		\item   $U_i = \{u_i^1, u_i^2, \dots, u_{i}^{\bar m}\}$ with $u_i^j \in \mathbb{R}^{m_i}, j \in \{1,\dots, \bar m\}$, is the control input set of dt-CS $\Upsilon\!_{i}$;
		\item $W_i \!\subseteq \mathbb R^{p_i}$ is the disturbance input set of dt-CS $\Upsilon\!_{i}$; 
		\item $f_i\!:\!X_i{\times} U_i{\times} W_i \!\rightarrow \!X_i$ is the transition map that characterizes the evolution of dt-CS $\Upsilon\!_{i}$, which is \emph{unknown in our setting}.
	\end{itemize}
	The evolution of dt-CS $\Upsilon\!_i$ can be described by a difference equation, expressed as 
	\begin{equation}\label{Eq_1a}
		\Upsilon\!_i\!:x_i(k+1) = f_i(x_i(k),u_i(k),w_i(k)),\quad 
		k\in \mathbb N,
		\quad
	\end{equation}
	where $x_i\!:\mathbb{N} \rightarrow X_i$, $u_i:\mathbb{N} \rightarrow U_i$, and  $w_i\!:\mathbb{N} \rightarrow W_i$ are the state, control and disturbance input signals, respectively.
\end{definition}

\BS{We now impose the following assumption on the transition map $f_i$, which is required throughout the paper.}
\begin{assumption}\label{assump:Lip_f}
		\BS{For each subsystem $\Upsilon_i = (X_i, U_i, W_i, f_i)$, the transition map $f_i(x_i,u_i,w_i)$ is Lipschitz continuous with respect to $(x_i, w_i)$ for every fixed $u_i \in U_i$. That is, there exists $L_i \in \mathbb{R}^+$ satisfying
		\[
		\| f_i(x_i, u_i, w_i) - f_i(x_i^\prime, u_i, w_i^\prime) \| \le L_i \big(\|x_i - x_i^\prime\| + \|w_i - w_i^\prime\|\big),
		\]
		for all $x_i, x_i^\prime \in X_i$, $w_i, w_i^\prime \in W_i$, and $u_i \in U_i$.}
\end{assumption}
\begin{remark}
	The primary rationale for treating $U_i$ as a finite set stems from the prevalent application of \emph{digital controllers} in practical scenarios. In addition, the disturbance input $w_i$ captures the influence from neighboring subsystems connected to a particular subsystem within the interconnection topology, acting as an \emph{adversarial} input.
\end{remark}

The subsequent definition formally delineates interconnected networks, formed by interconnecting disturbance inputs of individual subsystems.

\begin{definition} \label{network}
	Consider dt-CS $\Upsilon\!_{i}=(X_i,U_i,W_i,f_i), i\in\{1,\dots,M\}$, 
	with the following interconnection constraint:
	\begin{align}\label{Eq:4}
		\BS{[w_{1}; \dots; w_M] = g(x_1,\dots, x_M),}
	\end{align}
	where $g: \prod_{i = 1}^M X_i \to  \prod_{i = 1}^M W_i$ is an interconnection map.  Then, an interconnected network can be expressed by the tuple $\Upsilon = (X, U, f)$, where $X\coloneq\prod_{i=1}^{M}X_i$,  $U\coloneq\prod_{i=1}^{M}U_i$, and $f \coloneq [f_1;\dots;f_M]$\BS{, with $f$ inherently incorporating the influence of the unknown interconnection map $g$ in~\eqref{Eq:4}}.
	The interconnected network, denoted as $\Upsilon = \mathscr{N}\!(\Upsilon_1,\ldots,\Upsilon_M)$, operates according to
	\begin{equation}\label{Eq_1a1}
		\Upsilon\!:x(k+1) = {f(x(k),u(k))}, \quad k\in\mathbb N,
	\end{equation}
	\BS{where $x : \mathbb{N} \rightarrow X$ and $u : \mathbb{N} \rightarrow U$ represent the state and control signals of the interconnected network $\Upsilon$ in~\eqref{Eq_1a1}, which are defined as
		\(
		x \coloneq [x_1;\dots;x_M] \text{ and } u \coloneq [u_1;\dots;u_M].
		\)}
     We refer to a sequence $x_{x_0 u}\!\!: \mathbb N \rightarrow X$ that satisfies~\eqref{Eq_1a1} as the \textit{state trajectory} of $\Upsilon$, starting from an initial state $x_0$, subject to an input trajectory $u(\cdot)$.
\end{definition}

\begin{remark}
In our framework, we allow the interconnection topology $g$ in \eqref{Eq:4} to be \emph{unknown}, a common scenario in many real-world applications. Notably, in the small-gain results presented in \cite{lavaei2023symbolic}, this topology must not only be known but also specifically constrained as $w_{ij} = x_j$ for any $i, j \in \{1, \dots, M\}, i \neq j$, with $w_{ij}$ being partition elements of $w_i$ (see \cite[(3) and (4)]{lavaei2023symbolic}). In contrast, our work generalizes this constraint to an \emph{unknown} interconnection map $g$, which is not restricted to any specific form \BS{(cf. Section~\ref{subsec:HN_Case})}.
\end{remark}

\subsection{Symbolic Models}
Considering that the original subsystems evolve in a \emph{continuous-space} domain, analyzing them poses significant challenges. To alleviate this, we approximate a subsystem $\Upsilon_{\!i}$, $i\in\{1,\dots,M\}$, with a symbolic model characterized by \emph{discrete} state and disturbance input sets~\cite{pola2016symbolic}. The approximation algorithm initially partitions continuous state and disturbance input sets into finite segments, denoted as $X_i=\cup_j \mathsf{X}_i^j$ and $W_i=\cup_j \mathsf{W}_i^j$, respectively, and then selects representative points $\hat{x}_i ^j\in \mathsf{X}_i^j$ and $\hat{w}_i ^j\in \mathsf{W}_i^j$ as discrete states and disturbance inputs. In the following definition, we formally introduce how a symbolic model can be constructed.

\begin{definition}\label{def:sym}
	Given a dt-CS $\Upsilon\!_i=(X_i, U_i, W_i, f_i),  i\in\{1,\dots,M\}$, the symbolic model $\hat{\Upsilon}\!_i$ can be constructed as
	\begin{align}\label{symbolic}
		\hat{\Upsilon}\!_i=(\hat{X}_i, U_i, \hat W_i, \hat{f}_i),
	\end{align}
	where $\hat{X}_i\coloneq\big\{\hat{x}_i^j \mid j=1, \ldots, n_{\hat{x}_i}\big\}$ and $\hat{W}_i\coloneq\big\{\hat{w}_i^j \mid j=1, \ldots, n_{\hat{w}_i}\big\}$ are discrete state and disturbance input sets of $\hat{\Upsilon}\!_i$, respectively. Moreover, $\hat{f}_i: \hat{X}_i \times U_i \times \hat W_i \rightarrow \hat{X}_i$ is defined as
	\begin{align}
		\hat{f}_i(\hat{x}_i, u_i, \hat w_i)=\mathcal{Q}_i(f_i(\hat{x}_i, u_i, \hat w_i)),\label{eq:4.5}
	\end{align}
	where the quantization map $\mathcal{Q}_i: X_i \rightarrow \hat{X}_i$ allocates to any $x_i \in X_i$ and $w_i \in W_i$ a representative point $\hat{x}_i \in \hat{X}_i$ and $\hat{w}_i \in \hat{W}_i$ of the corresponding partition set, and satisfies the inequality
	\begin{align}\label{EQ:4}
		\|\mathcal{Q}_i(x_i)-x_i\| \leq \delta_i, \quad \forall x_i \in X_i,
	\end{align}
	with $\delta_i\coloneq\sup \Big\{\big\|x_i-x_i^{\prime}\big\|, x_i, x_i^{\prime} \in \mathsf{X}_i^j, j=1,2, \ldots, n_{\hat{x}_i}\!\Big\}$ being the \emph{state discretization parameter}.
\end{definition}

Since symbolic models in~\eqref{symbolic} operate in a discrete domain, algorithmic techniques from computer science can be applied to synthesize controllers enforcing complex logical properties~\cite{baier2008principles}. The primary concern is how to transfer  properties of interest demonstrated by symbolic models to the original systems. To accomplish this, it is necessary to establish a \emph{similarity relation} between state trajectories of two systems, employing the concept of alternating (sub-)bisimulation functions, as detailed in the next section.

\section{Alternating (Sub-)Bisimulation Functions} \label{barrier}
We initially establish the concept of alternating sub-bisimulation function between a dt-CS and its symbolic model incorporating disturbance inputs, as outlined in the following definition~\cite{tabuada2009verification,swikir2019compositional}.
\begin{definition} \label{Def:41}
	Consider a dt-CS $\Upsilon\!_i=(X_i,U_i, W_i,f_i)$ and its symbolic model $\hat{\Upsilon}\!_i=(\hat{X}_i, U_i, \hat W_i, \hat{f}_i)$. A function $\mathcal{V}_i:X_i\times \hat X_i\rightarrow\mathbb{R}_{0}^{+}$ is an \emph{alternating sub-bisimulation function (ASBF)} between $\hat \Upsilon\!_i$ and $\Upsilon\!_i$, denoted by $\hat{\Upsilon}_{\! i} \cong_{{\mathcal{V}_i}} \!\!\Upsilon_{\! i} $, if the following conditions hold:
	\begin{subequations}\label{6}
		\begin{itemize}
			\item $\forall x_i \in X_i,\forall \hat x_i \in \hat X_i\!:$
			\begin{align}
				\alpha_i\Vert x_i-\hat x_i\Vert^2 \leq \mathcal V_i(x_i,\hat x_i), \label{Eq:8_21}
			\end{align}
			\item $\forall x_i \in X_i,\forall \hat x_i \in \hat X_i,\forall u_i \in U_i, \forall w_i \in W_i,\forall \hat w_i \in \hat W_i\!:$
			\begin{align}
				&\mathcal V_i\big(f_i(x_i,u_i,w_i),\hat f_i(\hat x_i,u_i,\hat w_i)\big)\notag\\& \leq\;\gamma_i\mathcal V_i(x_i,\hat x_i)+\rho_i\Vert w_i-\hat w_i\Vert^2 + \psi_i, \label{Eq:8_22}
			\end{align}
		\end{itemize}
	\end{subequations}
	for some $\alpha_i,\psi_i  \in \mathbb{R}^{+}$, $\rho_i\in\mathbb{R}_{0}^{+}$, and $\gamma_i\in (0,1)$.
\end{definition}

\begin{remark}
	The ASBF establishes a similarity relation between the state trajectories of each subsystem and its symbolic model. In essence, if the states of $\Upsilon_{\! i}$ and $\hat{\Upsilon}_{\! i}$ begin from two close initial conditions, captured by $\mathcal{V}_i(x_i,\hat{x}_i)$ in~\eqref{Eq:8_21}, they will maintain proximity as their dynamics evolve over the subsequent time horizon, captured by the right-hand side of~\eqref{Eq:8_22}~\cite{tabuada2009verification}.
\end{remark} 

We now introduce a similar notion to establish a relation between two networks \emph{without} disturbance inputs.
\begin{definition} \label{cbc}
	Consider an interconnected network $\Upsilon = (X, U, f)$ and its symbolic model $\hat{\Upsilon}=(\hat X, U, \hat f)$. A function $\mathcal V:X\times \hat X \rightarrow \mathbb{R}_0^+$ is referred to as an \emph{alternating bisimulation  function (ABF)} between $\hat \Upsilon$ and $\Upsilon$, denoted by $\hat{\Upsilon} \cong_{\mathcal{V}} \!\!\Upsilon$, if there exist  $\alpha,\psi \in \mathbb{R}^{+}$, and $\gamma\in (0,1)$, such that the subsequent conditions hold:
	\begin{subequations}
		\begin{itemize}
			\item $\forall x \in X, \forall \hat x \in \hat X\!:$
			\begin{align}
				\alpha\Vert x-\hat x\Vert ^2\leq \mathcal V(x,\hat x), \label{alpha12}
			\end{align}
			\item  $\forall x \in X, \forall \hat x \in \hat X, \forall u\in U\!:$
			\begin{align}
				\mathcal V\big(f(x,u), \hat f(\hat x, u)\big) \leq \gamma\mathcal V(x,\hat x) + \psi. \label{alpha1}
			\end{align}
		\end{itemize}
	\end{subequations}
\end{definition}
\begin{remark}\label{rem:existence}
	\BS{According to~\cite{pola2009symbolic, swikir2019compositional}, the existence of an ASBF between a subsystem $\Upsilon_i$ and its symbolic model $\hat{\Upsilon}_i$ can be ensured if $\Upsilon_i$ is incrementally input-to-state stabilizable ($\delta$-ISS); see~\cite{angeli2002lyapunov, tran2016incremental} for more details. While the conditions of incremental ISS Lyapunov functions and ASBFs are similar, the main distinction lies in the state spaces on which they are defined; incremental ISS Lyapunov functions are defined over $X_i \times X_i$, whereas ASBFs are defined over $X_i \times \hat X_i$ (cf.~Definition~\ref{Def:41}), thereby relating trajectories of $\Upsilon_{\! i}$ and $\hat{\Upsilon}_{\! i}$. 
		In the particular case of linear dynamics, where incremental ISS coincides with conventional stabilizability, an ASBF exists whenever $\Upsilon_{\! i}$ is stabilizable. We note that once ASBFs are established for all subsystems, an ABF between the overall interconnected network and its symbolic model can be compositionally derived, provided that the sufficient compositionality condition holds (cf. Theorem~\ref{Thm1}).}
\end{remark}

\BS{In the following, we formally define an $\epsilon$-approximate alternating bisimulation relation between $\hat{\Upsilon}$ and $\Upsilon$, which is subsequently used to quantify the distance between their state trajectories~\cite[Definition~14]{girard2007approximation}.}

\begin{definition}\label{def:eps_rel}
	\BS{
	Consider an interconnected network $\Upsilon = (X, U, f)$ and its symbolic model $\hat{\Upsilon} = (\hat{X}, U, \hat{f})$. 
	A relation $\mathscr{R} \subseteq X \times \hat{X}$ is defined as an $\epsilon$-approximate alternating bisimulation relation between $\hat{\Upsilon}$ and $\Upsilon$ if, for all $(x, \hat{x}) \in \mathscr{R}$, the following conditions hold: %\ALL{Put the numbers in (i) ...}
	\begin{itemize}
		\item[(i)] $\|x - \hat{x}\| \le \epsilon$;
		\item[(ii)] For all $u \in U$, let $x^\prime = f(x,u)$ and $\hat{x}^\prime = \hat{f}(\hat{x},u)$. Then $(x^\prime,\, \hat{x}^\prime) \in \mathscr{R}$.
	\end{itemize}
	}
\end{definition}

\BS{Definition~\ref{def:eps_rel} implies that the interconnected network $\Upsilon$ and its symbolic model $\hat{\Upsilon}$ can replicate each other’s trajectories up to a deviation bounded by $\epsilon$. Building on Definition~\ref{def:eps_rel}, the following theorem employs an ABF to formally quantify the distance between the state trajectories of an interconnected network $\Upsilon$ and its symbolic model $\hat{\Upsilon}$.}
\begin{theorem}\label{thm-J19}
	Consider an interconnected network $\Upsilon = (X, U, f)$ and its symbolic model $\hat{\Upsilon} = (\hat{X}, U, \hat{f})$. 
	Suppose $\mathcal{V}$ is an ABF between $\hat{\Upsilon}$ and $\Upsilon$ as in Definition~\ref{cbc}. 
	Then, the relation $\mathscr{R} \subseteq X \times \hat{X}$ defined by
	\begin{align}\label{relation}
		\mathscr{R} := \{ (x, \hat{x}) \in X \times \hat{X} \mid \mathcal{V}(x, \hat{x}) \leq \bar{\psi} \},
	\end{align}
	is an $\epsilon$-approximate alternating bisimulation relation between $\hat{\Upsilon}$ and $\Upsilon$\BS{, as in Definition~\ref{def:eps_rel}}, with
	\begin{align}\label{error}
		\epsilon = \left(\frac{\bar{\psi}}{\alpha}\right)^{\!\!\frac{1}{2}}\!, 
		\quad \text{where} \quad 
		\bar{\psi} = \frac{\psi}{(1 - \gamma)\eta},
	\end{align}
	for any $0 < \eta < 1$.
\end{theorem}

\begin{proof}
	\BS{The proof comprises two parts based on Definition~\ref{def:eps_rel}: 
		(i) showing that $\|x-\hat x\|\le \epsilon$ for every $(x,\hat x)\in\mathscr R$, and 
		(ii) for any $(x,\hat x)\in\mathscr R$ and any $u\in U$, showing that $(x^\prime,\hat x^\prime)\in\mathscr R$, where $x^\prime = f(x,u)$ and 
		$\hat x^\prime = \hat f(\hat x,u)$.}
	The first part can be shown according to condition~\eqref{alpha12} and the relation $\mathscr{R}$ in~\eqref{relation}: 
	\begin{align*}
		\alpha\Vert x-\hat x\Vert ^2\leq \mathcal V(x,\hat x) \leq \bar \psi  \quad \to \quad \Vert x-\hat x\Vert \leq \big(\frac{\bar\psi}{\alpha}\big)^{\!\frac{1}{2}} = \epsilon.
	\end{align*}
	We now proceed with showing the second item. Since
	\begin{align*}
		\mathcal V\big(f(x,u), \hat f(\hat x, u)\big) \leq \gamma\mathcal V(x,\hat x) + \psi \leq \max\{\bar\gamma\mathcal V(x,\hat x), \bar\psi \},
	\end{align*}
	with $\bar\gamma$ and $\bar\psi$ as
	\begin{align*}
		\bar\gamma = 1 - (1 - \eta)(1 - \gamma),\quad \bar\psi = \frac{\psi}{(1 - \gamma)\eta},
	\end{align*} 
	for any~ $0<\eta <1$, one has $\mathcal{V}(x', \hat{x}') \leq \bar\psi$ given that $\bar \gamma\in(0,1)$ and $\mathcal{V}(x, \hat{x}) \leq \bar\psi$ according to~\eqref{relation}, indicating that $(x',\hat x') \in \mathscr{R}$, which concludes the proof.
\end{proof}

\begin{remark}
	Since Theorem \ref{thm-J19} guarantees closeness between the state trajectories of an interconnected network and its symbolic model, the underlying findings can be utilized to enforce various complex properties over the interconnected network, such as safety, reachability, and reach-while-avoid \cite{tabuada2009verification}. In particular, a symbolic model can facilitate the enforcement of such properties in simplified discrete-space domains and the refinement of the results back to the complex original systems while maintaining a quantifiable error bound on their closeness, as specified in \eqref{error}; all of which is made possible by the power of Theorem \ref{thm-J19}.
\end{remark}
\begin{remark}\label{rem:connection_dissi_small}
	\BS{The conditions in~\eqref{6}–\eqref{error} are not substitutes for small-gain conditions; rather, they pertain to ASBFs, ABFs, and the $\epsilon$-approximate alternating bisimulation relation. The only conceptual relation between our work and the small-gain results in the literature is that, in approaches grounded in small-gain reasoning, the right-hand side of condition~\eqref{Eq:8_22} typically includes the term $\rho_i\Vert w_i-\hat{w}_i\Vert^2$ to capture the influence of the disturbance input, which is also the case in our work. If one is interested in extending our results to a dissipativity-based reasoning setting, the term $\rho_i\Vert w_i-\hat w_i\Vert^2$ in~\eqref{Eq:8_22} should be replaced by
		\(
		\begin{bmatrix}
			w_i\\
			x_i
		\end{bmatrix}^\top\!
		{\underbrace{\begin{bmatrix}
					\mathds Z^{11}_i& \mathds  Z^{12}_i\\
					\mathds Z^{21}_i& \mathds  Z^{22}_i
			\end{bmatrix}}_{\mathds Z_i}}\!\begin{bmatrix}
			w_i\\
			x_i
		\end{bmatrix}\!\!,
		\)
		where $\mathds Z_i$ is a symmetric block matrix, which needs to be designed.}
\end{remark}

Generally, establishing a similarity relation \BS{using an ABF} between a large-scale network and its symbolic model introduces a significant challenge due to its  computationally intensive nature. To tackle this, our \emph{divide-and-conquer} strategy entails considering each network as an interconnection of individual subsystems with smaller dimensions, focusing solely on the subsystem level by constructing ASBFs \emph{from data}, and then integrating them to form an ABF for the interconnected network under an innovative \emph{data-driven compositionality} condition. It is worth noting that constructing an ASBF directly for each subsystem faces obstacles due to the appearance of \emph{unknown models} in condition \eqref{Eq:8_22}.

Inspired by the underlying challenges, we now formally define the problem that we aim to solve in this work.

\begin{resp}
	\begin{problem}\label{Problem: linear}
		Consider an interconnected network $\Upsilon = \mathscr{N}\!(\Upsilon_1,\ldots,\Upsilon_M)$ with an \emph{unknown} interconnection topology, comprising an \emph{arbitrary, a priori undefined} number of agents $\Upsilon\!_i$, each with an \emph{unknown dynamics} $f_i$. Develop a \emph{direct data-driven divide-and-conquer strategy} for constructing an ABF between $\hat \Upsilon$ and $\Upsilon$ based on ASBFs of individual subsystems $\Upsilon_{\! i}$ and $\hat{\Upsilon}_{\! i}$, while providing provable correctness guarantees. Accordingly, synthesize a controller that guarantees desired behaviors over the unknown network by employing local controllers derived from symbolic models of its subsystems.
	\end{problem}
\end{resp}

In the following section, we introduce our data-driven approach to address Problem~\ref{Problem: linear}.

\section{\BS{Data-Driven Construction of Symbolic Models and ASBFs}}\label{Data-Driven}
\BS{In this section, we propose our data-driven framework for the construction of ASBFs. Specifically, Section~\ref{sub:data_symbolic} presents the data-driven procedure for constructing symbolic models, which is subsequently used for the data-driven construction of ASBFs in Section~\ref{sub:data_ASBFs}.}

\subsection{\BS{Data-Driven Construction of Symbolic Models}}\label{sub:data_symbolic}
\BS{For simplicity, we present a simple example to describe the construction of the symbolic model $\hat f_i(\hat x_i, u_i, \hat w_i)$ and express each step through its corresponding mathematical equation. For brevity, we drop the subscript $i$ in this subsection, indicating that we describe the data-driven construction of the symbolic model for a single subsystem; the procedure for the remaining subsystems is similar.}

\BS{To this end, consider a system with the state vector $x = [x_1;\,  x_2]$, where the state set is defined as $X = X_1 \times X_2$, with $X_1$ and $X_2$ corresponding to the state variables $x_1$ and $x_2$, respectively. Similarly, let the disturbance input set be $W = W_1 \times W_2$, where $W_1$ and $W_2$ represent the disturbance inputs $w_1$ and $w_2$, respectively, with the disturbance input vector given by $w = [w_1; \, w_2]$. Assume that the system has a scalar control input $u$, with $U = \{u_1, u_2, \dots, u_{\bar m}\}$ denoting the associated discrete control input set.}

\BS{To construct the symbolic model $\hat f(\hat x, u, \hat w)$ in a data-driven fashion, we first overlay a grid on the state and disturbance input sets. Since the control input set $U$ is already discrete, no discretization is required for it. Consequently, by considering the grid with the state discretization parameter $\delta$ (cf.~\eqref{EQ:4}), the discrete state and disturbance input sets $\hat X = \hat X_1 \times \hat X_2$ and $\hat W = \hat W_1 \times \hat W_2$ are obtained, where $\hat X_1 = \{\hat x_1^1, \dots, \hat x_1^{n_{\hat x_1}}\}$ and $\hat X_2 = \{\hat x_2^1, \dots, \hat x_2^{n_{\hat x_2}}\}$, with $n_{\hat x_1}$ and $n_{\hat x_2}$ being the cardinalities of the discrete sets $\hat X_1$ and $\hat X_2$, respectively. Likewise, we have $\hat W_1 = \{\hat w_1^1, \dots, \hat w_1^{n_{\hat w_1}}\}$ and $\hat W_2 = \{\hat w_2^1, \dots, \hat w_2^{n_{\hat w_2}}\}$, with $n_{\hat w_1}$ and $n_{\hat w_2}$ being the cardinalities of the discrete sets $\hat W_1$ and $\hat W_2$, respectively. Thus, one obtains $\hat{X} = \big\{\hat{x}^j \mid j=1, \ldots, n_{\hat x} \big\}$, with $n_{\hat x} = n_{\hat x_1} n_{\hat x_2}$ and $\hat x^{j} = [\hat x^{j_1}_1;\, \hat x^{j_2}_2], j_1\in\{1,\dots,n_{\hat x_1}\}, j_2\in\{1,\dots,n_{\hat x_2}\}$. Similarly, we have $\hat{W} = \big\{\hat{w}^j \mid j=1, \ldots, n_{\hat w} \big\}$, with $n_{\hat w}  = n_{\hat w_1} n_{\hat w_2}$ and $\hat w^{j} = [\hat w^{j_1}_1;\, \hat w^{j_2}_2], j_1\in\{1,\dots,n_{\hat w_1}\}, j_2\in\{1,\dots,n_{\hat w_2}\}$. For ease of notation, we drop the superscript $j$ from $\hat x^{j}$ and $\hat w^{j}$ in the sequel.}

\begin{figure}[t!]
	\centering
	\includegraphics[width=0.9\linewidth]{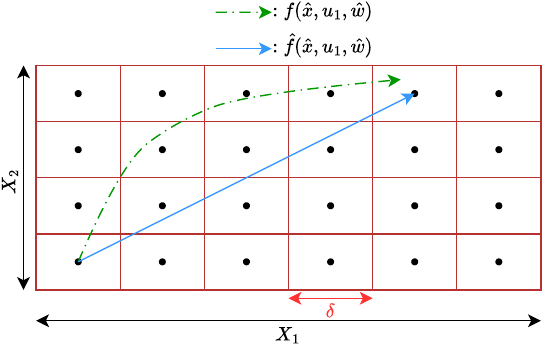}
	\caption{\BS{Illustration of the procedure for constructing the symbolic model. The center of each cell is taken as a representative point, corresponding to an element of $\hat X$.}}
	\label{fig:disc}
\end{figure}

\BS{For illustration, Figure~\ref{fig:disc} shows the grid overlaid on the state set. It is important to emphasize that the grid structure does not need to be rectangular; this particular choice is made solely for simplicity of implementation and visualization. Once these discrete sets are obtained, under a fixed control input $u_1$ selected from $U$, we initialize the system with $\hat x$ and $\hat w$ drawn from the discrete sets $\hat X$ and $\hat W$, respectively. Hence, it is clear that $f(\hat x, u_1, \hat w)$ is the black-box system output (\emph{i.e.}, the successor state) in this case.}

\BS{Having obtained $f(\hat x, u_1, \hat w)$, the next step is to determine $\hat f(\hat x, u_1, \hat w)$. Since $\hat f: \hat X \times U \times \hat W \rightarrow \hat X$, each cell center in Figure~\ref{fig:disc} (\emph{i.e.,} elements of the discrete set $\hat X$) can potentially represent the value of $\hat f$ corresponding to the computed $f(\hat x, u_1, \hat w)$. However, the valid $\hat f$, \emph{i.e.}, the appropriate cell center, is the one located $\delta$-close to $f(\hat x, u_1, \hat w)$. To put it mathematically, the valid $\hat f(\hat x, u_1, \hat w)$ is the one satisfying condition~\eqref{EQ:4} as
	\[
	\|\underbrace{\mathcal{Q}(f(\hat x, u_1, \hat w))}_{\hat f(\hat x, u_1, \hat w) ~ - ~\text{see}  ~\eqref{eq:4.5}} - f(\hat x, u_1, \hat w)\| \leq \delta,
	\]
	which corresponds precisely to the representative point closest to $f(\hat x, u_1, \hat w)$. This process is illustrated in Figure~\ref{fig:disc} for a single point. 
	It is crucial to emphasize that this procedure should be repeated for all possible combinations of $\hat x \in \hat X$, $\hat w \in \hat W$, and $u \in U$. For instance, for any state $\hat x \in \hat X$ and any control input $u_j \in U$, with $j \in \{1, 2, \dots, \bar m\}$, the value $\hat f(\hat x, u_j, \hat w)$ should be computed for every possible disturbance input $\hat w \in \hat W$. This process is iteratively continued until $\hat f(\hat x, u, \hat w)$ is determined for all combinations. We finally note that this method of constructing symbolic models from data introduces an approximation error, which can be captured by $\psi_i$ in condition~\eqref{Eq:8_22}.}

\subsection{\BS{Data-Driven Construction of ASBFs}}\label{sub:data_ASBFs}
Within our data-driven scheme, we consider the structure of ASBFs as
\begin{align}\label{Eq:3}
	\mathcal{V}_i(q_i,x_i,\hat x_i)=\sum_{j=1}^{r} {q}_i^jp_i^j(x_i,\hat x_i), 
\end{align}
where $p_i^j(x_i,\hat x_i), i\in\{1,\dots,M\}$, represent user-defined (nonlinear) basis functions, and  $q_i=[{q}_{i}^1;\ldots;q_i^r] \in \mathbb{R}^r$ denote unknown variables that need to be designed. For instance, if $p_i^j(x_i,\hat{x}_i)$ are selected as monomials with respect to $(x_i,\hat{x}_i)$, then $\mathcal{V}_i$ become a polynomial function. We assume each $\mathcal{V}_i$ is continuously differentiable.
\begin{remark}\label{rem:basis_functions_selection}
	\BS{The template in~\eqref{Eq:3} is general in the sense that users may select arbitrary continuous basis functions $p_i^j(x_i,\hat x_i)$, and the ASBF is constructed as a linear combination of these functions, ensuring that SOP~\eqref{SOP} remains linear in the decision variables $q_i$. For illustration, consider a network of scalar subsystems. For $i=1$, one may choose
		\(
		p_1(x_1,\hat x_1) = [1;\, x_1-\hat x_1;\, (x_1-\hat x_1)^2;\, \ldots;\, (x_1-\hat x_1)^d],
		\)
		where $d$ is a user-defined degree, leading to $\mathcal V_1(q_1,x_1,\hat x_1)=q_1^{\top}p_1(x_1,\hat x_1)$, with $q_1 = [{q}_{1}^1;\, {q}_{1}^2;\, {q}_{1}^3;\, \ldots;\, {q}_{1}^{d+1}]$.
		We note that the basis functions are not limited to polynomials; they may also include nonlinear terms such as $\sin(x_i-\hat x_i)$, $\cos(x_i\hat x_i)$, or $\ln(1+(x_i-\hat x_i)^2)$, depending on the subsystem characteristics. In practice, the selection of basis functions is user-defined and may require iterative refinement, which can be challenging in certain cases, and should be guided by physical intuition about the underlying (black-box) subsystem to fulfill~\eqref{eq: Compositional_Condition}. We note that since ASBFs serve as Lyapunov-like functions that provide only sufficient conditions, failure to satisfy~\eqref{eq: Compositional_Condition} with a chosen template does not imply the nonexistence of an ABF.}
\end{remark}

We commence by reformulating the requisite conditions for constructing ASBFs, as stated in Definition~\ref{Def:41}, into the following robust optimization program (ROP):
\begin{mini!}|s|[2]<b>
	{\left[\mathcal G_i;\mu_i;\varpi_i\right]}{\mu_i + \varpi_i,}{\label{ROP}}\notag
	{}
	\addConstraint{\forall (x_i,u_i,w_i) \in X_i \times U_i \times W_i}{}\notag
	\addConstraint{\forall (\hat x_i, u_i,\hat w_i) \in \hat X_i \times  U_i \times \hat W_i}{}\notag
	\addConstraint{- \mathcal V_i(q_i,x_i,\hat x_i)}{\leq \mu_i}\label{ROP0}
	\addConstraint{\alpha_i{\Vert x_i-\hat x_i\Vert^2} - \mathcal V_i(q_i,x_i,\hat x_i)}{\leq \mu_i}\label{ROP3}
	\addConstraint{\pmb{\mathcal{R}}}{\leq \mu_i}\label{ROP4}
	\addConstraint{\rho_i\Vert w_i-\hat w_i\Vert^2}{\leq \varpi_i}\label{ROP5}
	\addConstraint{\mathcal G_i = [\alpha_i;\gamma_i;\rho_i;\psi_i;{q}_{i}^1;\dots;q_{i}^r]}{}\notag
	\addConstraint{\alpha_i, \psi_i,\varpi_i \!\in\! \mathbb{R}^+,
		\rho_i\!\in\!\mathbb{R}_{0}^{+}, \mu_i \!\in\! \mathbb R,\gamma_i\!\in\! (0,1),}{}\notag
\end{mini!}
with
\begin{align*}
	\pmb{\mathcal{R}} = &\;\mathcal V_i\big(q_i, f_i(x_i,u_i,w_i),\hat f_i(\hat x_i, u_i,\hat w_i)\big) \! -\! \gamma_i\mathcal V_i(q_i,x_i,\hat x_i)\!\\
	&-\!\rho_i\Vert w_i \!-\!\hat w_i\Vert^2 \!-\! \psi_i.
\end{align*}
The optimal value of ROP \eqref{ROP} is denoted as $\mu^*_{R_i} + \varpi_{R_i}^*$\!.

\begin{remark}
	Note that conditions \eqref{Eq:8_21}-\eqref{Eq:8_22} do not originally include $\mu_i$ and $\varpi_i$. We amended these conditions within the ROP~\eqref{ROP} by incorporating the objective value $\mu_i + \varpi_i$ and will subsequently  leverage it in our compositional data-driven technique in \eqref{Con2}. It is apparent that any feasible solution to the ROP with $\mu^*_{R_i}  \leq 0$ confirms the satisfaction of conditions \eqref{Eq:8_21}-\eqref{Eq:8_22}.
\end{remark}
\begin{remark}
	The choice of $\mathcal{V}_i$ in~\eqref{Eq:3} enables the ROP to mainly exhibit convexity with respect to decision variables. The sole instance of scalar bilinearity arises in condition \eqref{ROP4} between $q_i^j$ and $\gamma_i$. To address this, we treat $\gamma_i\in(0,1)$ as an element of a finite set with a cardinality of $l$, denoted as $\gamma_i \in \{\gamma_{i}^1,\dots,\gamma_{i}^l\}$, and a-priori fix the choice of $\gamma_i$ before solving ROP~\eqref{ROP}.
\end{remark}

The proposed ROP in~\eqref{ROP} introduces two primary challenges, rendering it intractable. Firstly, the continuous nature of the spaces $X_i$ and $W_i$ leads to the presence of infinitely many constraints. Furthermore, the unknown maps $f_i$ appear in condition~\eqref{ROP4}, adding another layer of complexity to the problem. Given these significant obstacles, we propose a direct data-driven approach to construct ASBFs without directly solving ROP~\eqref{ROP}.
\BS{Before proceeding, we formalize the following assumption on the availability of measured data, which underpins the subsequent data-driven construction.}
\begin{assumption}\label{assump:measure}
	\BS{For each subsystem $\Upsilon_i = (X_i, U_i, W_i, f_i)$, the output is the identity map of the state variables, implying that all state variables of each subsystem are available as measurements.}
\end{assumption}

\BS{To proceed with the data-driven framework, we collect pairs of two consecutive samples} from trajectories of the unknown subsystems in the form of $((x_i^z, u_i, w_i^z), f_i(x_i^z, u_i, w_i^z))$, \BS{as enabled by Assumption~\ref{assump:measure}}, where $z\in\{1,\dots,\mathcal{N}_i\}$, \BS{with $\mathcal{N}_i$ denoting the number of collected samples}. Subsequently, we compute the maximum distance between $(x_i,u_i,w_i)\in X_i\times U_i \times W_i$ and the collected samples as follows:
\begin{align}
	\label{eq: eps}
	&\sigma_{i} = \max_{(x_i,u_i, w_i)}\min_z\Vert (x_i,u_i, w_i) - ( x^z_i, u_i, w_i^z)\Vert,\\
	& \quad\quad\forall (x_i,u_i,w_i)\in X_i\times U_i \times W_i.\notag
\end{align}\vspace{-0.5cm}
\begin{remark}\label{new8}
	The maximum distance $\sigma_{i}$ in \eqref{eq: eps} can be computed through grid-based partitioning of the space $(X_i\times U_i \times W_i)$. As computation with a grid-based approach reduces to a finite problem, the computational complexity scales linearly with both the number of samples and the number of grid points corresponding to the size of the grid.
\end{remark}

By considering $( x^z_i, u_i, w_i^z)\in X_i\times U_i \times W_i, z\in\{1,\dots,\mathcal N_i\}$, we now propose the following scenario optimization program (SOP):
\begin{mini!}|s|[2]<b>
	{\left[\mathcal G_i;\mu_i;\varpi_i\right]}{\mu_i + \varpi_i,}{\label{SOP}}\notag
	{}
	\addConstraint{\forall z\in\{1,\dots,\mathcal N_i\}}{}\notag
	\addConstraint{\forall (x_i^z,u_i,w_i^z) \in X_i \times U_i \times W_i}{}\notag
	\addConstraint{\forall (\hat x_i,u_i,\hat w_i) \in \hat X_i \times  U_i \times\hat W_i}{}\notag
	\addConstraint{- \mathcal V_i(q_i,x_i^z,\hat x_i)}{\leq \mu_i}\label{SOP0}
	\addConstraint{\alpha_i{\Vert x_i^z-\hat x_i\Vert^2} - \mathcal V_i(q_i,x_i^z,\hat x_i)}{\leq \mu_i}\label{SOP3}
	\addConstraint{\pmb{\mathcal{R}^z}}{\leq \mu_i}\label{SOP4}
	\addConstraint{\rho_i\Vert w_i^z-\hat w_i\Vert^2}{\leq \varpi_i}\label{SOP5}
	\addConstraint{\mathcal G_i = [\alpha_i;\gamma_i;\rho_i;\psi_i;{q}_{i}^1;\dots;q_{i}^r]}{}\notag
	\addConstraint{\alpha_i, \psi_i,\varpi_i \!\in\! \mathbb{R}^+,
		\rho_i\!\in\!\mathbb{R}_{0}^{+}, \mu_i \!\in\! \mathbb R,\gamma_i\!\in\! (0,1),}{}\notag
\end{mini!}
with
\begin{align*}
	\pmb{\mathcal{R}^z} = &\; \mathcal V_i\big(q_i, f_i(x_i^z,u_i,w_i^z),\hat f_i(\hat x_i, u_i,\hat w_i)\big) - \gamma_i\mathcal V_i(q_i,x_i^z,\hat x_i)\\
	& -\rho_i\Vert w_i^z-\hat w_i\Vert^2 - \psi_i.
\end{align*}

It is evident that the proposed SOP in~\eqref{SOP} contains a finite number of constraints, all of which have the same form as those in~\eqref{ROP}. Additionally, the challenge posed by unknown $f_i$ and $\hat{f}_i$ is addressed, as both can be derived from data \BS{(cf.~Section~\ref{sub:data_symbolic})}. The optimal value of SOP is denoted by $\mu_{\mathcal N_i}^* + \varpi_{\mathcal N_i}^*$.
\begin{remark}
	The term $\rho_i\Vert w^z_i-\hat w_i\Vert^2$ will be integrated into the compositional data-driven condition we introduce in the subsequent section (see Theorem~\ref{Thm1}). That was the primary motivation behind incorporating an additional condition into~\eqref{ROP5} and~\eqref{SOP5}, where the objective entails minimizing both $\mu_i$ and $\varpi_i$.
\end{remark}
\begin{remark}\label{rem:feas}
	\BS{While ROP~\eqref{ROP} and SOP~\eqref{SOP} encode the same inequalities yielding ASBFs, the former represents the \emph{model-based} formulation under known subsystem dynamics, whereas the latter serves as its data-driven counterpart, evaluated over sampled data. Both programs are always feasible in the soft sense since $\mu_i \in \mathbb{R}$ acts as a slack variable and can take sufficiently large positive values to satisfy the constraints over compact domains. However, such feasibility does not ensure that the obtained solution yields a valid ABF. In fact, a valid ABF can only be guaranteed if the data-driven compositional condition~\eqref{eq: Compositional_Condition} is satisfied.}
\end{remark}

In the upcoming section, we leverage our proposed SOP and introduce a \emph{newly developed compositional data-driven} condition to construct an ABF for an interconnected network based on the ASBFs of its individual subsystems.

\section{Data-Driven Construction of  ABFs}\label{Guarantee_ROP}
To offer our data-driven findings, we make the assumption that each unknown map $f_i$ is Lipschitz continuous with respect to $(x_i, w_i)$, for any $(\hat x_i, u_i,\hat w_i) \in \hat X_i \times  U_i \times \hat W_i$, an assumption that holds true in numerous practical scenarios \BS{(cf.~Assumption~\ref{assump:Lip_f})}. Given that ASBF $\mathcal V_i$ is continuously differentiable and since our analysis is conducted on the bounded domain $X_i \times U_i \times W_i$, it follows that $\mathcal V_i^*(q_i,f_i(x_i,u_i,w_i),\hat f_i(\hat x_i,u_i,\hat w_i)) - \gamma_i^*\mathcal V_i^*(q_i,x_i,\hat x_i)$, with $\mathcal V_i^*(q_i,\cdot,\cdot) = \mathcal V_i(q_i^*,\cdot,\cdot)$ is also Lipschitz continuous with respect to  $(x_i, w_i)$, for any $(\hat x_i, u_i,\hat w_i) \in \hat X_i \times  U_i \times \hat W_i$, with a Lipschitz constant denoted as $\mathscr{L}^2_{i}$. Likewise, employing the same reasoning, one can demonstrate that $\alpha_i^*{\Vert x_i-\hat x_i\Vert^2} - \mathcal V_i^*(q_i,x_i,\hat x_i)$, is also Lipschitz continuous with respect to $x_i$,  for any $\hat x_i\in \hat X_i$, with a Lipschitz constant $\mathscr{L}^1_{i}$. At a later stage, we present an approach to compute Lipschitz constants $\mathscr{L}^1_{i}, \mathscr{L}^2_{i}$ from the collected data (cf. Algorithm~\ref{Alg:1}).
\begin{remark}\label{rem:why_lip}
	\BS{As our analysis is conducted over a bounded domain, and since $\hat x_i \in \hat X_i$, where $\hat X_i$ is a (finite) discrete set, \emph{i.e.}, $\hat x_i$ only takes discrete values, the function $\Omega_i(q_i, x_i, \hat x_i) = \alpha_i^*\Vert x_i-\hat x_i\Vert^2 - \mathcal V_i^*(q_i, x_i, \hat x_i)$, where $\mathcal V_i^*$ is continuously differentiable, is Lipschitz continuous with respect to $x_i$, for any $\hat x_i \in \hat X_i$, with a Lipschitz constant $\mathscr{L}^1_i$. More precisely, if one aims to compute this constant analytically, one has
		\begin{align}
			\mathcal{L}_{i} = \max_{x_i \in X_i} \Big\Vert \frac{\partial \Omega_i(q_i, x_i, \hat x_i)}{\partial x_i} \Big\Vert. \label{eq:lips}
		\end{align}
		It is evident that $\Vert \frac{\partial \Omega_i(q_i, x_i, \hat x_i)}{\partial x_i} \Vert$ depends not only on $x_i$ but also on $\hat x_i$, which only takes discrete values. This implies that, to obtain a single $\mathscr{L}^1_i$ applicable to all $\hat x_i \in \hat X_i$, one should compute $\mathcal{L}_i$ based on~\eqref{eq:lips} for all possible discrete values of $\hat x_i \in \hat X_i$. 
		Specifically, for each discrete value of $\hat x_i$, one obtains a corresponding value $\mathcal{L}_{i_j}$, where $j \in \{1, 2, \dots, n_{\hat x_i}\}$ and $n_{\hat x_i}$ denotes the cardinality of the set $\hat{X}_i$. Consequently, the unified Lipschitz constant $\mathscr{L}^1_i$ adopted for all $\hat x_i \in \hat X_i$ is computed by
		\(
		\mathscr{L}^1_i = \max\{\mathcal{L}_{i_1}, \mathcal{L}_{i_2}, \dots, \mathcal{L}_{i_{n_{\hat x_i}}}\!\}.
		\)
		The same reasoning can be applied to obtain $\mathscr{L}^2_i$ as well.}
\end{remark}

We now offer the following theorem, as the main result of the work, facilitating the construction of an ABF for an interconnected network with an unknown interconnection topology based on the ASBFs of its individual subsystems, while providing correctness guarantees. \BS{The Lipschitz constants $\mathscr{L}^1_{i}$ and $\mathscr{L}^2_{i}$, required in the following theorem, are computed according to Algorithm~\ref{Alg:1}, in line with the rationale provided in Remark~\ref{rem:why_lip}.}

\begin{theorem}\label{Thm1}
	Consider an interconnected network $\Upsilon = \mathscr{N}\!(\Upsilon_1,\ldots,\Upsilon_M)$, composed of \emph{an arbitrary, a priori undefined} number of agents $\Upsilon\!_i$, with a \emph{fully-unknown interconnection topology}. Consider the $\text{SOP}$ in~\eqref{SOP} for individual subsystems with its corresponding optimal value $\mu_{\mathcal N_i}^* + \varpi_{\mathcal{N}_i}^*$ and solution $\mathcal G_i^* = [\alpha_i^*;\gamma_{i}^*;\rho_i^*;\psi_i^*;{q}^{1*}_{i};\dots;q^{r*}_{i}]$. If 
	\begin{subequations}\label{eq: Compositional_Condition}
	    \BS{\begin{align}
        & \mu_{\mathcal N_i}^* + \mathscr{L}_{i} \sigma_{i} \leq 0, \quad \forall i \in \{1, \ldots, M\}, \label{Con1}\\
		&\sum_{i=1}^M\big(\mu_{\mathcal N_i}^* + \varpi_{\mathcal N_i}^* + \mathscr{L}_{i} \sigma_{i} \big) \leq 0,\label{Con2}
	\end{align}}
	\end{subequations}
	with $\mathscr{L}_{i} = \max\{\mathscr{L}^1_{i},\mathscr{L}^2_{i}\}$, then 
	\begin{align}\label{Lyp}
		\mathcal V(q,x,\hat x):= \sum_{i=1}^M\mathcal V_i^*(q_i,x_i,\hat x_i),
	\end{align}
	is an ABF between $\hat\Upsilon$ and $\Upsilon$, \emph{i.e.,} $\hat \Upsilon \cong_{\mathcal{V}} \!\! \Upsilon$, with a correctness guarantee, satisfying conditions~\eqref{alpha12}-\eqref{alpha1} with $\gamma = \underset{i}\max\{\gamma_i^*\}, \alpha = \underset{i}\min\{\alpha_i^*\},\text{ and } \psi = \sum_{i=1}^M\psi_i^*.$
\end{theorem}

\begin{proof}
	We first show that under \BS{condition~\eqref{Con1}}, the ABF $\mathcal V$ in~\eqref{Lyp} satisfies condition~\eqref{alpha12} for the whole range of the state set,
	\emph{i.e.,}
	\begin{align}\label{new}
		\alpha\Vert x-\hat x\Vert^2 \leq \mathcal V(q,x,\hat x), \quad \quad \forall x \in X,\:\forall\hat x\in \hat X\!.
	\end{align}
	Given the proposed form of the ABF, one has:
	\begin{align*}
		&\alpha\Vert x-\hat x\Vert^2 - \mathcal V(q,x,\hat x) = \alpha\Vert x-\hat x\Vert^2 - \sum_{i=1}^M \mathcal V_i^*(q_i,x_i,\hat x_i)\\
		& = \alpha\Vert [x_1;\dots;x_{M}]-[\hat x_1;\dots;\hat x_{M}]\Vert^2 - \sum_{i=1}^M \mathcal V_i^*(q_i,x_i,\hat x_i)\\
		& \;\BS{= \alpha\sum_{i=1}^M\Vert x_i-\hat x_i\Vert^2} - \sum_{i=1}^M\mathcal V_i^*(q_i,x_i,\hat x_i)\\&= \sum_{i=1}^M\big(\alpha\Vert x_i-\hat x_i\Vert^2 - \mathcal V_i^*(q_i,x_i,\hat x_i)\big).
	\end{align*}
	By defining $\alpha = \underset{i}\min\{\alpha_i^*\}$, one has
	\begin{align*}
		&\alpha\Vert x-\hat x\Vert^2 - \mathcal V(q,x,\hat x)\\ &\leq\sum_{i=1}^M \big(\min_i\{\alpha_i^*\}\Vert x_i-\hat x_i\Vert^2 - \mathcal V_i^*(q_i,x_i,\hat x_i)\big)\\&\leq \sum_{i=1}^M \big(\alpha_i^*\Vert x_i-\hat x_i\Vert^2 - \mathcal V_i^*(q_i,x_i,\hat x_i)\big).
	\end{align*}
	Let us define $z^*:= \arg \, \underset{z}{\min} \Vert x^z_i - x_i\Vert.$  By incorporating the term $\sum_{i=1}^M(\alpha_i^*{\Vert  x^{z^*}_i-\hat x_i\Vert^2} - \mathcal V_i^\ast(q_i,x^{z^*}_i,\hat x_i))$ through addition and subtraction, we have
	\begin{align*}
		\alpha\Vert x-\hat x\Vert^2 - \mathcal V(q,x,\hat x)&\leq\sum_{i=1}^M\big(\alpha_i^*\Vert x_i-\hat x_i\Vert^2 - \mathcal V_i^*(q_i,x_i,\hat x_i)\\&\;\;-\alpha_i^*{\Vert x^{z^*}_i\!-\hat x_i\Vert^2} + \mathcal V_i^\ast(q_i, x^{z^*}_i,\hat x_i)\\
		& \;\;+ \alpha_i^*{\Vert x^{z^*}_i\!-\hat x_i\Vert^2} -  \mathcal V_i^\ast(q_i, x^{z^*}_i,\hat x_i)\big).
	\end{align*}
	Given that $\alpha_i^*{\Vert x_i-\hat x_i\Vert^2} - \mathcal V_i^*(q_i,x_i,\hat x_i)$ is Lipschitz continuous with respect to $x_i$,  for any $\hat x_i\in \hat X_i$, with a Lipschitz constant $\mathscr{L}^1_{i}$, since $\min_z\Vert x_i- x^z_i\Vert \leq \min_z\Vert (x_i,u_i, w_i) - ( x^z_i, u_i, w_i^z)\Vert$, $\mathscr{L}_{i} = \max\{\mathscr{L}^1_{i},\mathscr{L}^2_{i}\}$, and as per~\eqref{eq: eps}, we have
	\begin{align*}
		&\alpha\Vert x-\hat x\Vert^2 - \mathcal V(q,x,\hat x)\\
		&\leq\sum_{i=1}^M\big(\mathscr{L}^1_{i} \, \min_z\Vert x_i\!-\! x^z_i\Vert\!+\! \alpha_i^*{\Vert  x^{z^*}_i\!-\hat x_i\Vert^2} - \mathcal V_i^\ast(q_i, x^{z^*}_i\!\!,\hat x_i)\big)\\
		&\leq\sum_{i=1}^M\big(\mathscr{L}^1_{i} \, \min_z\Vert (x_i,\!u_i,\! w_i) \!-\! ( x^z_i,\! u_i,\! w_i^z)\Vert\!+\! \alpha_i^*{\Vert  x^{z^*}_i\!\!-\!\hat x_i\Vert^2} \!\\&\hspace{1.1cm}-\! \mathcal V_i^\ast(q_i, x^{z^*}_i\!,\!\hat x_i)\big)\\
		&\leq \sum_{i=1}^M\!\big(\mathscr{L}^1_{i}\!\!\! \max_{(x_i,u_i,w_i)}\!\!\min_z\!\Vert (x_i,\!u_i,\! w_i) \!-\! ( x^z_i,\! u_i,\! w_i^z)\Vert\!\\&\hspace{1.1cm}+\! \alpha_i^*{\Vert x^{z^*}_i\!\!\!-\!\hat x_i\Vert^2} \!-\! \mathcal V_i^\ast(q_i, x^{z^*}_i\!\!,\hat x_i)\big)\\
		& \leq \sum_{i=1}^M\big(\mathscr{L}_{i} \sigma_i\!+\! \alpha_i^*{\Vert x^{z^*}_i\!-\hat x_i\Vert^2} \!-\! \mathcal V_i^\ast(q_i, x^{z^*}_i\!,\hat x_i)\big).
	\end{align*}
	According to condition~\eqref{SOP3} of SOP, we have
	\begin{align*}
		\alpha\Vert x-\hat x\Vert^2 - \mathcal V(q,x,\hat x)\leq\sum_{i=1}^M\big(\mu_{\mathcal N_i}^* + \mathscr{L}_{i} \sigma_{i}\big).
	\end{align*}
	Given the proposed \BS{condition in~\eqref{Con1}}, one can conclude that
	\begin{align*}
		 &\forall x \in X,\: \hat x\!\in\! \hat{X}\!,
		&\alpha\Vert x-\hat x\Vert^2 - \mathcal V(q,x,\hat x)\leq 0.
	\end{align*}
	\BS{We define $\mu^*_{R_i}$ in ROP~\eqref{ROP} as $\mu^*_{R_i} = \mu_{\mathcal N_i}^* + \mathscr{L}_{i} \sigma_{i}$, for any $i \in \{1, \ldots, M\}$. This expression is negative for all subsystems $\Upsilon_{\! i}$ under the condition in~\eqref{Con1}.} Under similar reasoning steps, one can show that 
	\begin{align*}
		\mathcal V(q,x,\hat x) \geq 0, \quad \quad \forall x \in X,\:\forall\hat x\in \hat X\!.
	\end{align*}
	We now proceed with showing that under condition~\eqref{Con2}, $\mathcal V$ also satisfies condition~\eqref{alpha1} for the whole range of the state set, \emph{i.e.,}
	\begin{align*}
		\mathcal V(q,f(x, u), \hat f(\hat x, u)) \!-\! \gamma\mathcal V(q,x,\hat x) - \psi \leq 0,  \forall x \in X, \: \hat x\in \hat X.
	\end{align*}
	Given the proposed form of the ABF in~\eqref{Lyp} and taking into consideration that $ \psi = \sum_{i=1}^M\psi_i^*$, one can write down $\mathcal V$ based on ASBFs of individual subsystems as:
	\begin{align*}
		&\mathcal V\big(q,f(x,\!u), \hat f(\hat x,\! u)\big) \!-\! \gamma\mathcal V(q,\!x,\!\hat x) \!-\! \psi\\&=\!\sum_{i=1}^M\!\!\big(\mathcal V_i^*\!(q_i,\!f_i(x_i,\!u_i,\!w_i),\!\hat f_i(\hat x_i,\!u_i,\!\hat w_i)) \!-\!  \gamma\mathcal V_i^*\!(q_i,\!x_i,\!\hat x_i)\!-\!\psi_i^*\big).
	\end{align*}
	By defining $\gamma = \underset{i}\max\{\gamma_i^*\}$, one has
	\begin{align*}
		&\mathcal V\big(q,f(x,u), \hat f(\hat x, u)\big)\! - \!\gamma\mathcal V(q,x,\hat x)-\psi\\
		&=\sum_{i=1}^M\big(\mathcal V_i^*(q_i,f_i(x_i,u_i,w_i),\hat f_i(\hat x_i,u_i,\hat w_i)) \\&\hspace{1.1cm}- \underset{i}\max\{\gamma_i^*\}\mathcal V_i^*(q_i,x_i,\hat x_i)-\psi_i^*\big)\\
		&\leq\sum_{i=1}^M\big(\mathcal V_i^*(q_i,f_i(x_i,u_i,w_i),\hat f_i(\hat x_i,u_i,\hat w_i))\\&\hspace{1.1cm} -  \gamma_i^*\mathcal V_i^*(q_i,x_i,\hat x_i)-\psi_i^*\big).
	\end{align*}
	Let $$z^*:= \arg \, \underset{z}{\min} \Vert (x_i,u_i,w_i)- ( x^z_i,u_i, w^z_i)\Vert.$$  By incorporating the terms $$\sum_{i=1}^M\!(\mathcal V_i^\ast\!(q_i,f_i(x^{z^*}_i,u_i, w^{z^*}_i),\hat f_i(\hat x_i,u_i,\hat w_i)) \!-\! \gamma_i^*\mathcal V_i^\ast\!(q_i, x^{z^*}_i,\hat x_i))$$ through addition and subtraction, one has
	\begin{align*}
		&\mathcal V\big(q,f(x,u), \hat f(\hat x, u)\big)- \gamma\mathcal V(q,x,\hat x)-\psi\\
		&\leq \sum_{i=1}^M\big(\mathcal V_i^*(q_i,f_i(x_i,u_i,w_i),\hat f_i(\hat x_i,u_i,\hat w_i)) -  \gamma_i^*\mathcal V_i^*(q_i,x_i,\hat x_i)\\
		& ~~~~~~-\!\!\,\psi_i^*\\
		&~~~~~~~\!\! -\!\!\mathcal V_i^\ast\!(q_i,f_i(x^{z^*}_i,u_i, w^{z^*}_i),\hat f_i(\hat x_i,u_i,\hat w_i)\!)\!+\!\!\, \gamma_i^*\mathcal V_i^\ast\!(q_i, x^{z^*}_i\!\!,\hat x_i)\\
		& ~~~~~~+\!\!\mathcal V_i^\ast(q_i,f_i(x^{z^*}_i\!\!,u_i, w^{z^*}_i),\hat f_i(\hat x_i,u_i,\hat w_i)\!)\!-\! \gamma_i^*\mathcal V_i^\ast(q_i, x^{z^*}_i\!\!,\hat x_i)\!\big).
	\end{align*}
	Given that $\mathcal V_i^*(q_i,f_i(x_i,u_i,w_i),\hat f_i(\hat x_i,u_i,\hat w_i)) - \gamma_i^*\mathcal V_i^*(q_i,x_i,\hat x_i)$ is Lipschitz continuous with respect to  $(x_i, w_i)$, for any $(\hat x_i, u_i,\hat w_i) \in \hat X_i \times  U_i \times \hat W_i$, with a Lipschitz constant $\mathscr{L}^2_{i}$, since $\mathscr{L}_{i} = \max\{\mathscr{L}^1_{i},\mathscr{L}^2_{i}\}$, and as per~\eqref{eq: eps}, we have
	\begin{align*}
		&\mathcal V\big(q,f(x,u), \hat f(\hat x, u)\big)-\gamma\mathcal V(q,x,\hat x)-\psi\\
		&\leq\sum_{i=1}^M\big(\mathscr{L}^2_{i} \, \min_z\Vert (x_i,u_i,w_i)- ( x^z_i, u_i, w^z_i)\Vert\\
		&\hspace{1.1cm} + \mathcal V_i^\ast(q_i,f_i(x^{z^*}_i\!\!,u_i, w^{z^*}_i),\hat f_i(\hat x_i,u_i,\hat w_i)\!) \\&\hspace{1.1cm}- \gamma_i^*\mathcal V_i^\ast(q_i, x^{z^*}_i\!\!,\hat x_i) - \psi_i^*\big)\\
		&\leq \sum_{i=1}^M\big(\mathscr{L}^2_{i} \!\!\!\max_{(x_i,u_i, w_i)}\!\! \min_z\Vert (x_i,u_i,w_i)- ( x^z_i, u_i, w^z_i) \Vert\\
		&\hspace{1.1cm} + \mathcal V_i^\ast(q_i,f_i(x^{z^*}_i\!\!,u_i, w^{z^*}_i),\hat f_i(\hat x_i,u_i,\hat w_i)\!) \\&\hspace{1.1cm}- \gamma_i^*\mathcal V_i^\ast(q_i, x^{z^*}_i\!\!,\hat x_i) - \psi_i^*\big)\\
		& \leq \sum_{i=1}^M\big(\mathscr{L}_{i} \sigma_{i} \!+\! \mathcal V_i^\ast(q_i,f_i(x^{z^*}_i\!\!,u_i, w^{z^*}_i\!),\hat f_i(\hat x_i,u_i,\hat w_i)\!) \\&\hspace{1.1cm}- \gamma_i^*\mathcal V_i^\ast(q_i, x^{z^*}_i\!\!,\hat x_i) \!-\! \psi_i^*\big).
	\end{align*}
	According to conditions~\eqref{SOP4} and~\eqref{SOP5} of SOP, we have
	\begin{align*}
		&\mathcal V\big(q,f(x,u), \hat f(\hat x, u)\big)\! - \!\gamma\mathcal V(q,x,\hat x) - \psi\\&\leq \sum_{i=1}^M\big(\mu_{\mathcal N_i}^* + \varpi_{\mathcal N_i}^* + \mathscr{L}_{i} \sigma_i \big).
	\end{align*}
	Under the proposed condition in~\eqref{Con2}, one can conclude that
	\begin{align*}
		&\forall x \in X, \: \hat x\in \hat X,\\
		&\mathcal V(q,f(x, u), \hat f(\hat x, u)) \leq  \gamma\mathcal V(q,x,\hat x) + \psi.
	\end{align*}
	Then, $\mathcal V (q, x, \hat{x})$ in the form of~\eqref{Lyp} is an ABF between $\hat{\Upsilon}$ and $\Upsilon$, \textit{i.e.}, $\hat \Upsilon \cong_{\mathcal{V}} \!\! \Upsilon$, with a correctness guarantee, concluding the proof.
\end{proof}

\begin{algorithm}[t!]
	\caption{Data-driven ABF construction with provable guarantees}\label{Alg:2}
	\begin{center}
		\begin{algorithmic}[1]
			\REQUIRE 
			Degree of ASBF $\mathcal{V}_i$
			\FOR{all unknown subsystems $\Upsilon\!_{i} , i\in\{1,\dots,M\}$,}
			\STATE 
			Gather a set of two-consecutive sampled data from trajectories of each unknown subsystem $\Upsilon_{\! i}$ 
			\STATE
			Construct symbolic model $\hat \Upsilon_{\! i}$ according to Definition~\ref{def:sym} and \BS{Section~\ref{sub:data_symbolic}}
			\STATE
			Solve SOP~\eqref{SOP} to obtain ASBF $\mathcal{V}_i^\ast$, and $\gamma_i^\ast$, $\alpha_i^\ast$, $\psi_i^\ast$, $\rho_i^\ast$, $\varpi_i^\ast$, $\mu_i^\ast$
			\STATE
			Compute Lipschitz constant $\mathscr{L}_{i}$ according to Algorithm~\ref{Alg:1} \label{line4}
			\STATE
			Compute $\sigma_i$ according to~\eqref{eq: eps} and Remark~\ref{new8}
			\ENDFOR
			\IF{\BS{$\mu_{\mathcal N_i}^* \! + \! \mathscr{L}_{i} \sigma_{i} \! \leq \! 0$ and} $\sum_{i=1}^M \! \big(\mu_{\mathcal N_i}^* \! + \! \varpi_{\mathcal N_i}^* \! + \! \mathscr{L}_{i} \sigma_{i} \big) \! \leq \! 0$}
			\STATE
			$\mathcal V = \sum_{i=1}^M\mathcal V_i^\ast$ is an ABF between $\hat \Upsilon$ and $\Upsilon$ satisfying conditions~\eqref{alpha12}-\eqref{alpha1} with $\gamma = \underset{i}\max\{\gamma_i^*\}$, $\alpha = \underset{i}\min\{\alpha_i^*\}$, and $\psi = \sum_{i=1}^M\psi_i^*$
			\ELSE
			\STATE
			Return to Step 1, collect a larger set of data\footnotemark, or change the degree of ASBF $\mathcal{V}_i$, then repeat Steps 2-7
			\ENDIF
			\ENSURE
			ABF $\mathcal{V}$ between $\hat \Upsilon$ and $\Upsilon$ with provable guarantees
		\end{algorithmic}
	\end{center}
\end{algorithm}
\footnotetext{Essentially, a larger set of data implies a smaller $\sigma_i$, which can potentially ease the satisfaction of~\eqref{eq: Compositional_Condition}.}

	\begin{algorithm}[t]
	\caption{Estimation of Lipschitz constant $\mathscr{L}_{i}$ using data}\label{Alg:1}
	\begin{center}
		\begin{algorithmic}[1]
			\REQUIRE 
			ASBF $\mathcal{V}_i^\ast, \alpha_i^*, \gamma_i^*$
			\STATE 
			Choose $\bar{\mathscr{K}}, \tilde{\mathscr{K}} \in \mathbb{N}$ and $\mathscr{B} \in \mathbb{R}^+$
			\FOR{$\forall \hat x \in \hat X$, $\forall u \in U$, $\forall \hat w \in \hat W$}
			\FOR{$\theta \gets 1$ to $\tilde{\mathscr{K}}$}
			\FOR{$z \gets 1$ to $\bar{\mathscr{K}}$}
			\STATE
			Collect  sampled pairs $((x_i^z,w_i^z), (x_i^{z^\prime},w_i^{z^\prime}))$ such that $\Vert  (x_i^z,w_i^z) - (x_i^{z^\prime},w_i^{z^\prime}) \Vert \leq \mathscr{B}$
			\STATE
			Compute the slope $\kappa_i^z = \dfrac{\Vert  \mathscr{G}(x_i^z,u_i,w_i^z, \hat x_i,\hat w_i ) - \mathscr{G}(x_i^{z^\prime},u_i,w_i^{z^\prime}, \hat x_i,\hat w_i )  \Vert}{\Vert   (x_i^z,w_i^z ) - (x_i^{z^\prime},w_i^{z^\prime}) \Vert}$ with $\mathscr{G}(x_i^z,u_i,w_i^z, \hat x_i,\hat w_i) = \mathcal V_i^\ast(q_i,f_i(x_i^z,u_i,w_i^z),\hat f_i(\hat x_i,u_i,\hat w_i)) - \gamma_i^*\mathcal V_i^\ast(q_i,x_i^z,\hat x_i)$, ($\mathscr{G}(x_i^{z^\prime},u_i,w_i^{z^\prime}, \hat x_i,\hat w_i)$ is obtained similarly)
			\ENDFOR
			\STATE
			Obtain the maximum slope as $\mathcal{M}_\theta^i = \max \{\kappa_i^1,\dots,\kappa_i^{\bar{\mathscr{K}}}\}$ 
			\ENDFOR
			\ENDFOR
			\STATE
			Employing the \textit{Reverse Weibull distribution} on $\mathcal{M}_1^i, \dots, \mathcal{M}_{\tilde{\mathscr{K}}}^i$, which provides location, scale, and shape parameters, designate the \textit{location parameter} as an estimate of $\mathscr{L}^2_{i}$
			\STATE Repeat Steps 1-11 with the following $\mathscr{G}$ to estimate $\mathscr{L}^1_{i}$: \begin{align*}
				\mathscr{G}(x_i^z, \hat x_i) = \alpha_i^*{\Vert x_i^z-\hat x_i\Vert^2} - \mathcal V_i^*(q_i,x_i^z,\hat x_i)
			\end{align*}
			\ENSURE $\mathscr{L}_{i} = \max\{\mathscr{L}^1_{i},\mathscr{L}^2_{i}\}$ 
		\end{algorithmic}
	\end{center}
\end{algorithm}

\begin{remark}\label{rem:Lip_why}
	\BS{The Lipschitz constants $\mathscr{L}^1_i$ and $\mathscr{L}^2_i$ in our framework allow us to provide \emph{out-of-sample} performance guarantees, \emph{i.e.}, generalizing the solution obtained using finite data to unseen data. To discuss their effect, we note that the larger the value of $\mathscr{L}_{i} = \max\{\mathscr{L}^1_{i},\mathscr{L}^2_{i}\}$, the more challenging it becomes to satisfy~\eqref{eq: Compositional_Condition}, as higher Lipschitz constants intuitively indicate more aggressive system behavior. Moreover, we note that while our framework does not rely on model-based small-gain reasoning, it exhibits an analogous compensatory property: in~\eqref{Con2}, even if the quantity $\big(\mu_{\mathcal N_i}^* + \varpi_{\mathcal N_i}^* + \mathscr{L}_i\sigma_i\big)$ is positive for certain subsystems, its effect can be compensated by other subsystems for which this term is sufficiently negative, thereby preserving the overall compositional validity. Finally, unlike classical small-gain reasoning, which explicitly depends on the network topology, our data-driven approach implicitly captures the influence of interconnections within the collected data used to solve SOP~\eqref{SOP}.}
\end{remark}

We provide Algorithm~\ref{Alg:2}  that summarizes required steps  for establishing an ABF between the interconnected network and its symbolic model using data with provable guarantees.

{\bf Controller design process.} Due to the finite nature of symbolic models, algorithmic techniques from computer science can be leveraged to synthesize controllers that enforce complex logical specifications. This process entails designing local controllers for the symbolic models $\hat\Upsilon_i$, for $i \in \{1,\dots,M\}$, and then refining them back to the original subsystems $\Upsilon_i$ \BS{using the similarity relations established by the ASBFs. In particular, once the symbolic models and their corresponding ASBFs are constructed and the compositional condition~\eqref{eq: Compositional_Condition} is verified, the controller synthesis process becomes specification-independent. This means that, by leveraging software tools developed within the formal methods community, \emph{e.g.}, \texttt{SCOTS}~\cite{rungger2016scots}, one can synthesize controllers for the constructed symbolic models to enforce different complex specifications of interest, and subsequently transfer the obtained results back to the original continuous-space systems while guaranteeing the satisfaction of the same specifications. This feature contrasts with most conventional control methodologies, where the controller design process is tightly coupled with the specification of interest. This \emph{decoupling} between abstraction and specification represents a key flexibility of our approach compared with other control design methods.}
Finally, the controller for an interconnected network $\Upsilon = \mathscr{N}(\Upsilon_1,\ldots,\Upsilon_M)$ will be a vector, where each component corresponds to a controller for individual subsystems $\Upsilon_i$.

To verify the proposed data-driven compositional condition~\eqref{eq: Compositional_Condition}, knowledge of the Lipschitz constant $\mathscr{L}_{i}$ is required. To estimate it for each subsystem $\Upsilon_{\! i}$ using a finite dataset, we employ the fundamental result of \cite{wood1996estimation} and offer Algorithm~\ref{Alg:1} for its computation. Under this algorithm, the convergence of the estimated value $\mathscr{L}_{i}$ to its true value is assured in the limit, as supported by the following lemma~\cite{wood1996estimation}.
\begin{lemma}\label{lemma:Lip}
	The estimated $\mathscr{L}_{i}$ converges to its actual value if and only if $\mathscr{B}$ tends to zero, while $\bar{\mathscr{K}}$ and $\tilde{\mathscr{K}}$ approach infinity.
\end{lemma}

\begin{remark}
	To estimate the Lipschitz constant $\mathscr{L}_i$ in Algorithm~\ref{Alg:1}, one needs to determine unknown coefficients $q_i$, which requires solving the SOP~\eqref{SOP}. To avoid the need for subsequent verification of condition~\eqref{eq: Compositional_Condition}, one can assume a certain range for unknown coefficients $q_i$ and estimate the Lipschitz constant $\mathscr{L}_i$ before solving SOP~\eqref{SOP}. These established ranges should be then enforced during the solution of SOP~\eqref{SOP}.
\end{remark}

\section{\BS{Limitations and Potential Future Directions}}
\BS{Our findings also have some limitations, which we elaborate on in this section. As mentioned in Assumption~\ref{assump:Lip_f}, our framework is applicable only when the transition map $f_i(x_i(k),u_i(k),w_i(k))$ is Lipschitz continuous with respect to $(x_i, w_i)$. Even though we emphasize that this assumption is standard in the relevant literature and holds in many practical scenarios, we acknowledge that it may restrict the applicability of our framework in certain cases; for instance, our framework cannot be employed when the transition map $f_i(x_i(k),u_i(k),w_i(k))$ is discontinuous~\cite{cortes2008discontinuous}. Moreover, as established in Lemma~\ref{lemma:Lip}, the estimates of the Lipschitz constants $\mathscr{L}^1_{i}$ and $\mathscr{L}^2_{i}$, obtained through Algorithm~\ref{Alg:1}, converge to their true values \emph{asymptotically}. Developing a \emph{quantitative characterization} for estimating these constants from a finite amount of data remains a promising direction for future research.}

\BS{Another potential limitation in practice pertains to the data requirements for solving the proposed SOP~\eqref{SOP}. Specifically, we overlay a grid on the sample space to ensure uniform coverage during data collection and collect pairs of two consecutive samples from trajectories of the unknown subsystems in the form $((x_i^z, u_i, w_i^z), f_i(x_i^z, u_i, w_i^z))$, where $z \in \{1,\dots,\mathcal{N}_i\}$. This requirement necessitates initializing each subsystem \emph{multiple times} to cover the whole range of the sample space, which can be substantial in certain cases, thereby making the data collection process challenging in practice. In practical implementations, such data can be efficiently generated using high-fidelity simulators, which considerably simplify the data acquisition process. Despite this, to mitigate this limitation, a complementary direction arises from another active line of research on direct data-driven control based on the so-called Willems et al.’s fundamental lemma~\cite{willems2005note}; see~\cite{de2019formulas,martin2023guarantees}. Within this framework, a single trajectory of input–state or input–output data from each subsystem suffices, thereby substantially easing the procedure of data collection~\cite{samari2024abstraction}. However, it is important to note that this alternative line of research generally confines the class of systems to \emph{specific structural forms} (such as input-affine or solely polynomial systems), whereas our framework only requires the transition map to be Lipschitz continuous (see the third case study with non-affine, highly complex dynamics). Hence, our framework remains applicable to a broader class of systems without imposing the structural constraints typically present in the aforementioned approaches. Moreover, within this alternative line of work, one should typically select an extensive library of functions capable of representing the true terms in the system dynamics, which demands considerable prior insight into the system and can be restrictive in practice. In contrast, our framework does not rely on such assumptions.}

\BS{As mentioned in Assumption~\ref{assump:measure}, our proposed framework requires direct measurability of all state variables of each subsystem. We note that most direct data-driven approaches have been developed under the ubiquitous assumption of direct measurability of all state variables (see \emph{e.g.,}~\cite{makdesi2023data, hashimoto2022learning, lavaei2022data, lavaei2023symbolic}). Nevertheless, by employing a technique similar to that used in~\cite{jahanshahi2023data}, our framework could be extended to accommodate \emph{input-output} data instead of input-state data. In such an extension, the framework would be enhanced to account for both the unknown transition map $f_i(x_i(k),u_i(k),w_i(k))$ and the \emph{unknown output map} $h_i(x_i(k))$, where $y_i(k) = h_i(x_i(k))$ denotes the measured output. While this modification could potentially improve the practicality of the framework, such an extension requires assuming the existence of an estimator that can estimate the state variables of each subsystem with a known upper bound on the estimation error (see~\cite[Assumption~1]{jahanshahi2023data}). This bound would then need to be incorporated into the analysis and the proposed guarantee.}

\BS{While the proposed framework cannot be applied to networks composed of subsystems with time-varying dynamics or subject to stochasticity \emph{e.g.}, process noise, one could build upon it to develop methods for constructing symbolic models of such networks. When the dynamics of each subsystem, and consequently, of the entire interconnected network, are time-varying, the ASBFs between symbolic models of subsystems and their corresponding subsystems should also be time-varying. This leads, in turn, to time-varying ABFs between symbolic models of networks and the networks themselves. Such a setting necessitates modifications to the $\epsilon$-approximate alternating bisimulation relation in~\eqref{relation} and the corresponding $\epsilon$ in~\eqref{error}. More importantly, SOP~\eqref{SOP} must also be adapted to incorporate these time-varying ASBFs. In the case of the existence of process noise in the dynamics of subsystems, the expectation operator with respect to the underlying noise should be incorporated into conditions~\eqref{Eq:8_22},~\eqref{alpha1}, and~\eqref{SOP4} in SOP~\eqref{SOP}. Subsequently, one should estimate empirical approximations of these expectations and, to preserve formal rigor, quantify the deviation between the expected and empirical values, \emph{e.g.}, using Chebyshev’s inequality~\cite{hernandez2001chebyshev}, to establish a confidence level. Moreover, in this case, the closeness guarantee would itself be probabilistic, introducing \emph{two layers of probability}; the inner one reflecting the probabilistic closeness guarantee, and the outer one corresponding to the confidence associated with the empirical approximation of the expected values. It is also worth highlighting that, in the stochastic case, the symbolic model is generalized to a \emph{finite Markov decision process} (MDP), which is typically represented by a stochastic matrix. Each component of this matrix denotes the probability of transitioning from one discrete cell to another, given the applied discrete inputs and disturbances, and accounting for the inherent randomness within the system.}

\BS{In the current setting, the interactions between neighboring subsystems are considered as the sole source of disturbance in the dynamics of each subsystem, denoted by $w_i$. Consequently, the present framework does not provide robustness against \emph{external} sources of disturbance. Nonetheless, the current setting can be extended to account for additional sources of disturbance within the dynamics of each subsystem, and consequently, within the network as a whole. To do so, the dynamics of each dt-CS $\Upsilon\!_i$, given in~\eqref{Eq_1a}, should be reformulated as
	\[
	\Upsilon\!_i\!:x_i(k+1) = f_i(x_i(k),u_i(k),w_i(k),d_i(k)),\quad k\in \mathbb N,
	\]
	where $d_i\!:\mathbb{N} \rightarrow D_i$ denotes the \emph{external} disturbance signal, and $D_i$ represents the external disturbance input set of the dt-CS $\Upsilon\!_i$. Since $d_i$ represents an exogenous signal whose source lies outside the network, it should explicitly appear in~\eqref{Eq_1a1} (unlike $w_i$, which captures interconnection effects and disappears in the overall network representation), as
	\[
	\Upsilon\!:x(k+1) = f(x(k),u(k),d(k)), \quad k\in\mathbb N,
	\]
	where $d \coloneq [d_1;\dots;d_M]$, with $M$ denoting the number of subsystems, and $D:=\prod_{i=1}^{M}D_i$. Subsequently, analogous to the continuous internal disturbance input set $W_i$, which is discretized as $\hat{W}_i := \big\{\hat{w}_i^j \mid j=1, \ldots, n_{\hat{w}_i}\big\}$, the same procedure should be applied to the external disturbance input, resulting in the discrete set $\hat{D}_i := \big\{\hat{d}_i^j \mid j=1, \ldots, n_{\hat{d}_i}\big\}$. Accordingly, the symbolic model $\hat{f}_i: \hat{X}_i \times U_i \times \hat{W}_i \times \hat{D}_i \rightarrow \hat{X}_i$ is defined as
	\begin{align*}
		\hat{f}_i(\hat{x}_i, u_i, \hat{w}_i, \hat{d}_i) = \mathcal{Q}_i(f_i(\hat{x}_i, u_i, \hat{w}_i, \hat{d}_i)),
	\end{align*}
	where $\mathcal{Q}_i$ denotes the quantization map satisfying~\eqref{EQ:4}. As is evident, to incorporate the external disturbance $d_i$, it is necessary to assume that its corresponding set is known, which is a standard assumption in the related literature.}

\BS{Now, it is required to refine Definitions~\ref{Def:41} and~\ref{cbc}, as well as Theorem~\ref{thm-J19}. More concretely, in Definitions~\ref{Def:41} and~\ref{cbc}, conditions~\eqref{Eq:8_22} and~\eqref{alpha1}, respectively, should be modified as follows:
	\begin{itemize}
		\item $\forall x_i \in X_i,\forall \hat x_i \in \hat X_i,\forall u_i \in U_i, \forall w_i \in W_i,\forall \hat w_i \in \hat W_i, \forall d_i \in D_i,\forall \hat d_i \in \hat D_i\!:$
		\begin{align*}
			&\mathcal V_i\big(f_i(x_i,u_i,w_i,d_i), \hat f_i(\hat x_i,u_i,\hat w_i,\hat d_i)\big)\\
			&\leq \gamma_i \mathcal V_i(x_i,\hat x_i) + \rho_i \Vert w_i - \hat w_i \Vert^2 + \varkappa_i \Vert d_i - \hat d_i \Vert^2 + \psi_i,
		\end{align*}
		for some $\psi_i  \in \mathbb{R}^{+}$, $\rho_i, \varkappa_i\in\mathbb{R}_{0}^{+}$, and $\gamma_i\in (0,1)$,
	\end{itemize}
	and
	\begin{itemize}
		\item $\forall x \in X, \forall \hat x \in \hat X, \forall u\in U, \forall d \in D, \forall \hat d \in \hat D\!:$
		\[
		\mathcal V\big(f(x,u,d), \hat f(\hat x, u, \hat d)\big)
		\leq \gamma \mathcal V(x,\hat x) + \varkappa \Vert d - \hat d \Vert^2 + \psi,
		\]
		for some $\psi \in \mathbb{R}^{+}$, $\varkappa \in \mathbb{R}_{0}^{+}$, and $\gamma\in (0,1)$, where $\hat D := \prod_{i=1}^{M} \hat D_i$.
	\end{itemize}
	It is important to emphasize that conditions~\eqref{Eq:8_21} and~\eqref{alpha12} remain unchanged.
	Consequently,~\eqref{relation} and~\eqref{error} in Theorem~\ref{thm-J19} should be modified as well. To do so, let us assume that there exists a $\nu \in \mathbb{R}^{+}$ such that $\Vert d - \hat d \Vert^2 \leq \nu$ for all $d \in D$ and $\hat d \in \hat D$. Then,~\eqref{relation} should be replaced by
	\[
	\mathscr{R} := \{(x, \hat{x}) \in X \times \hat{X} \mid \mathcal{V}(x, \hat{x}) \leq \BS{\max\{\bar{\varkappa} \nu, \bar{\psi}\}}\},
	\]
	with
	\begin{align}
		\epsilon = \Big(\frac{\max\{\bar{\varkappa} \nu, \bar{\psi}\}}{\alpha}\Big)^{\!\!\tfrac{1}{2}}\!\!,\label{eq:R3_Rel}
	\end{align}
	where
	\begin{align}
		\bar{\varkappa} = \frac{(1 + \eta_2)\eta_3}{(1 - \gamma)\eta_1}\varkappa, 
		\quad \bar{\psi} = \frac{(1 + \eta_2)\eta_3}{(1 - \gamma)(\eta_3 - 1)\eta_1\eta_2}\psi,\label{eq:tmp_barvarkappa}
	\end{align}
	for any $\eta_1, \eta_2 \in (0, 1)$ and $\eta_3 \in (1, 2)$. To demonstrate this, we have
	\begin{align*}
		\mathcal V\big(f(x,u, d), \hat f(\hat x, u, \hat d)\big) & \leq \gamma \mathcal V(x,\hat x) + \varkappa \Vert d - \hat d \Vert^2 + \psi\\& \leq \max\{\bar\gamma\mathcal V(x,\hat x), \bar{\varkappa} \nu, \bar\psi \},
	\end{align*}
	with $\bar{\varkappa}$ and $\bar{\psi}$ as in~\eqref{eq:tmp_barvarkappa}, and
	\(
	\bar\gamma = 1 - (1 - \eta_1)(1 - \gamma),
	\)
	for any $\eta_1, \eta_2 \in (0, 1)$ and $\eta_3 \in (1, 2)$, which shows the correctness of~\eqref{eq:R3_Rel}. We also have
	\(
	\alpha\Vert x-\hat x\Vert ^2\leq \mathcal V(x,\hat x) \leq \max\{\bar{\varkappa} \nu, \bar{\psi}\}.
	\)
	Thus, we get
	\[
	\Vert x-\hat x\Vert \leq \Big(\frac{\max\{\bar{\varkappa} \nu, \bar{\psi}\}}{\alpha}\Big)^{\!\frac{1}{2}} = \epsilon.
	\]
	ROP~\eqref{ROP} and SOP~\eqref{SOP} should be modified accordingly; however, performing these revisions needs thorough investigation. We leave this task for future research.
}

\section{Case Studies}\label{Compu}
We illustrate our data-driven findings by applying them to \BS{three} benchmarks: \emph{(i)} a building temperature network, \emph{(ii)} a vehicle network, and \emph{(iii)} \BS{a highly nonlinear network with nonlinear interconnection constraints.} \BS{All} cases involve unknown networks with an \emph{arbitrary, a priori undefined} number of agents and unknown interconnection topologies. We demonstrate that our compositional data-driven approach can provide safety guarantees for \BS{these three} unknown networks. All simulations were performed in Matlab on a MacBook Pro (Apple M2 Max with 32GB memory), with the SOP~\eqref{SOP} solved using Mosek\footnote{Mosek license was obtained from https://www.mosek.com.} solver.
 
 \subsection{Room temperature network}
 We exemplify our data-driven findings over a room temperature network comprising $> 10000$ rooms\footnote{Note that our compositional condition in \eqref{eq: Compositional_Condition} does not require prior knowledge of the number of subsystems; it can be verified for any \emph{arbitrary, a priori undefined} number of subsystems. This is more flexible compared to the model-based small-gain approach~\cite{swikir2019compositional} or its data-driven version \cite{lavaei2023symbolic}, which requires both the interconnection topology and the exact number of subsystems to satisfy the traditional compositional condition.}, characterized by unknown mathematical models and an unknown interconnection topology. Within this network, each room is equipped with a cooler and used to preserve specialized medications at low temperatures~\cite{meyer2017compositional}.

The dynamics of the network follow the subsequent difference equation:
 	\begin{align*}
 		\Upsilon : x(k+1) = A x(k) + \BS{\phi} T_c u(k) + \BS{\lambda} T_E,
 	\end{align*}
 	where the matrix $A$ is defined by its diagonal entries as $a_{ii} = 1 - 2 \BS{\beta} - \BS{\lambda} - \BS{\phi} u_i(k)$, for $i \in \{1, \ldots, M\}$, while the off-diagonal entries could be either $\BS{\beta}$ or zero, depending on the unknown interconnection topology. Here, the parameters $\BS{\beta}$, $\BS{\lambda}$, and $\BS{\phi}$ represent thermal exchange coefficients, corresponding to the heat transfer between adjacent rooms $i \pm 1$ and $i$, the interaction between room $i$ and the external environment, and the cooling effect within room $i$, respectively. Additionally, the variables $x(k) = \left[x_1(k); \ldots ; x_{M}(k)\right]$ and $T_E = \left[T_{e_1}; \ldots; T_{e_{M}}\right]$, where each external temperature $T_{e_i}$ is set to $-2^{\circ} C$ for all $i \in \{1, \ldots, M\}$. The cooler temperature is specified as $T_c = 5^{\circ} \mathrm{C}$.
 	To isolate the dynamics of each room individually, we have
 	\begin{align*}
 		\Upsilon_i: x_i(k+1) = &\; a_{ii} x_i(k) + \BS{\beta} w_i(k) + \BS{\phi} T_c u_i(k) + \BS{\lambda} T_{e_i}.
 	\end{align*}
 	 The network $\Upsilon$ is thus expressed as $\Upsilon = \mathscr{N}\!\left(\Upsilon_1, \ldots, \Upsilon_{M}\right)$. It is assumed that both the model for each room and the interconnection topology are unknown.
 
 The primary objective is to  construct ASBFs and their symbolic models by solving SOP~\eqref{SOP} and compositionally design an ABF, derived from data, via the result of Theorem~\ref{Thm1}. Consequently, we leverage the data-driven symbolic models to design controllers within the control input set $U_i = \{0,  1\}$ that regulate the temperature of each room (\emph{i.e.,} $x_i$) within a predetermined safe set $X_i = [-0.5, \: 0.5]$ while ensuring guaranteed correctness. 
 
 We set the structure of our ASBF as $\mathcal{V}_i(q_i, x_i, \hat{x}_i) = q_{i}^1(x_i - \hat{x}_i)^6 + q_{i}^2(x_i - \hat{x}_i)^4 + q_{i}^3(x_i - \hat{x}_i)^2 + q_{i}^4$. We follow required steps in Algorithm \ref{Alg:2} by collecting trajectories and computing $\sigma_{i} = 0.05$. Then, by solving SOP~\eqref{SOP} for all $i \in \{1, \dots, M\}$, the corresponding decision variables are obtained as
 \begin{align}
 	&\mathcal{V}_i^\ast(q_i,x_i,\hat{x}_i) = 0.4949(x_i - \hat{x}_i)^6 -0.25 (x_i - \hat{x}_i)^4\notag \\&\hspace{2.2cm}+ 0.001(x_i - \hat{x}_i)^2 + 0.8,\label{new99}\\\notag
 	& \mu_i^\ast = -0.0496, \quad \varpi_i^\ast = 10^{-6},
 \end{align}
 with a fixed $\gamma_{i}^\ast = 0.985$. We now compute $\mathscr{L}_{i} = \max\{0.9675,0.7359\}=0.9675$ according to Algorithm \ref{Alg:1}. Given that $\mu_i^\ast + \varpi_i^\ast + \mathscr{L}_{i}\sigma_i = -0.0012 \leq 0$ for all $i \in \{1, \dots, M\}$, according to Theorem~\ref{Thm1}, $\mathcal V(q, x, \hat x)= \sum_{i = 1}^{M}\mathcal V_i^\ast(q_i,x_i,\hat{x}_i)$, with $\mathcal V_i^*$ as in \eqref{new99}, is an ABF between the unknown room temperature network and its symbolic model. Figure~\ref{fig:room1} visually demonstrates that the data-driven ASBF is non-negative in the whole range of state space, while satisfying condition~\eqref{Eq:8_21}. The execution time for this example was approximately 18 seconds, with a memory consumption of about $315$ Megabits.
 
 Leveraging the constructed data-driven symbolic models, we now \emph{compositionally} design a controller for the interconnected network such that the controller forces the temperature of each room to be inside the safe set $X_i = [-0.5, \: 0.5]$. To this end, we initially synthesize a local controller for each room using its symbolic model via \texttt{SCOTS}~\cite{rungger2016scots}. We then refine this controller back over each unknown original room using data-driven ASBFs. Consequently, the controller for the room temperature network is a vector, with each component being a controller for individual rooms. Closed-loop state trajectories of three arbitrary rooms and their corresponding control inputs can be seen in Figure~\ref{fig:room}. As illustrated, the synthesized controller can force trajectories of unknown rooms to stay in the safe set $X_i = [-0.5, \: 0.5]$.
 
 \BS{It is worth noting that while the control input in Figure~\ref{fig:room} resembles a bang-bang controller, this similarity arises solely from the formulation of the case study, in which the control input set is defined as $U_i = \{0,1\}$, corresponding to a cooler that is either on or off. Consequently, the control input switches abruptly between these two extreme values. However, if the control input set were defined differently, \emph{e.g.}, $U_i = \{0, 0.25, 0.5, 1\}$, the resulting control input would not exhibit bang-bang behavior, which our setting readily accommodates.}

\begin{figure}[t!]
	\centering
	\subfloat[\centering Satisfaction of $\mathcal{V}_i^\ast(q_i,x_i,\hat{x}_i) \geq 0$.]{{\resizebox{0.8\linewidth}{!}{\begin{tikzpicture}
					\begin{axis}[colorbar, grid = major,
						xlabel = $x_i$, ylabel = $\hat x_i$, zlabel = $\mathcal{V}_i^* \left(q_i{, } x_i{, }\hat x_i\right)$]
						\addplot3
						[surf,faceted color=blue, view={30}{30},
						samples=40, scatter, mark = *,
						domain=-0.5:0.5,y domain=-0.5:0.5]
						{0.4949*(x-y)^6 - 0.25*(x-y)^4 + 0.001*(x-y)^2+0.8};					
					\end{axis}
	\end{tikzpicture}} }}%

	\subfloat[\centering Satisfaction of $\alpha_i^\ast\Vert x_i-\hat x_i\Vert^2 - \mathcal V_i^\ast(q_i,x_i,\hat x_i)\leq 0$.]{{\resizebox{0.8\linewidth}{!}{\begin{tikzpicture}
					\begin{axis}[colorbar, grid = major,
						xlabel = $x_i$, ylabel = $\hat x_i$, zlabel = condition~\eqref{Eq:8_21}]
						\addplot3
						[surf,faceted color=blue, view={30}{30},
						samples=40, scatter, mark = *,
						domain=-0.5:0.5,y domain=-0.5:0.5]
						{(x-y)^2 - (0.4949*(x-y)^6 - 0.25*(x-y)^4 + 0.001*(x-y)^2+0.8)};
						
					\end{axis}
	\end{tikzpicture}} }}%
	\caption{Data-driven ASBF is non-negative in the whole range of state space (a), while satisfying condition~\eqref{Eq:8_21} (b).}%
	\label{fig:room1}%
\end{figure}

 \begin{figure}[t!]
	\centering
	\subfloat[\centering Temperature of each room \label{subfig1}]{{\includegraphics[width=0.8\linewidth]{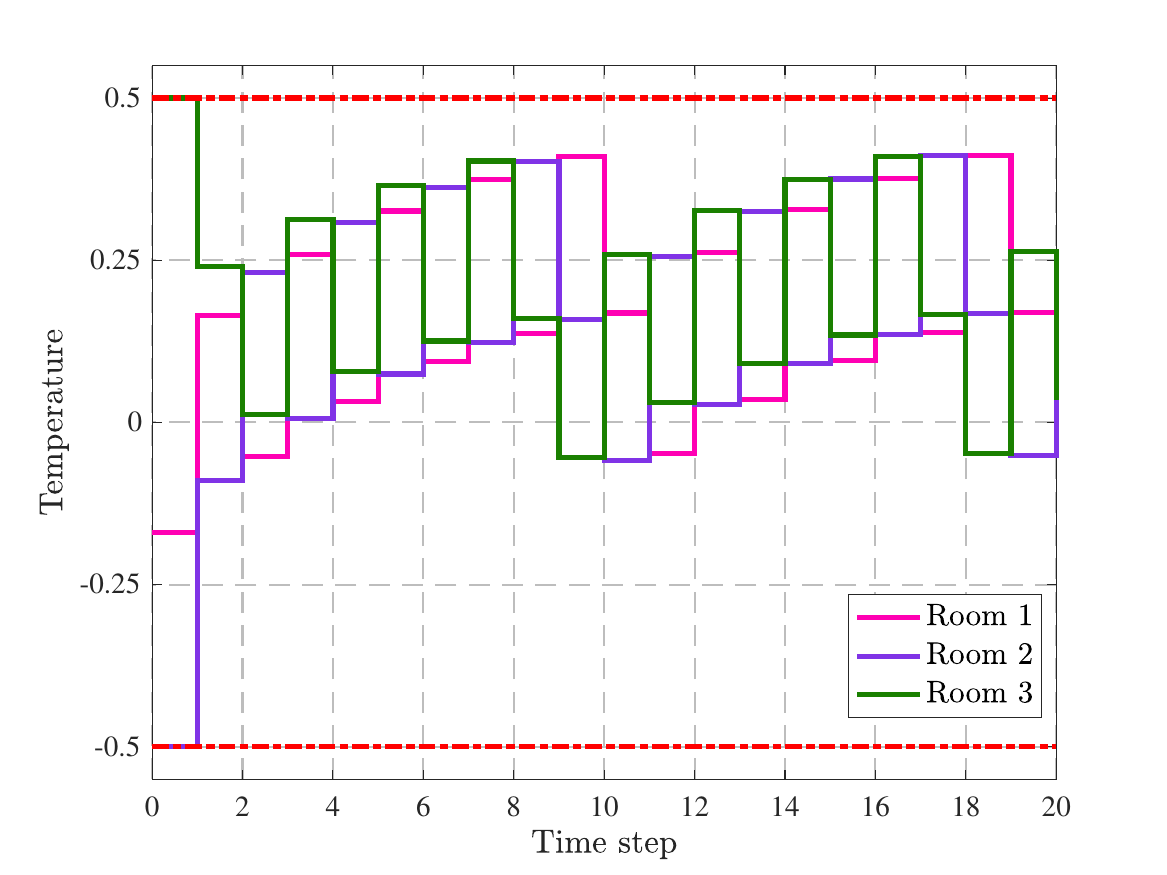}}}%
	
	\subfloat[\centering Control input of each room]{{\includegraphics[width=0.8\linewidth]{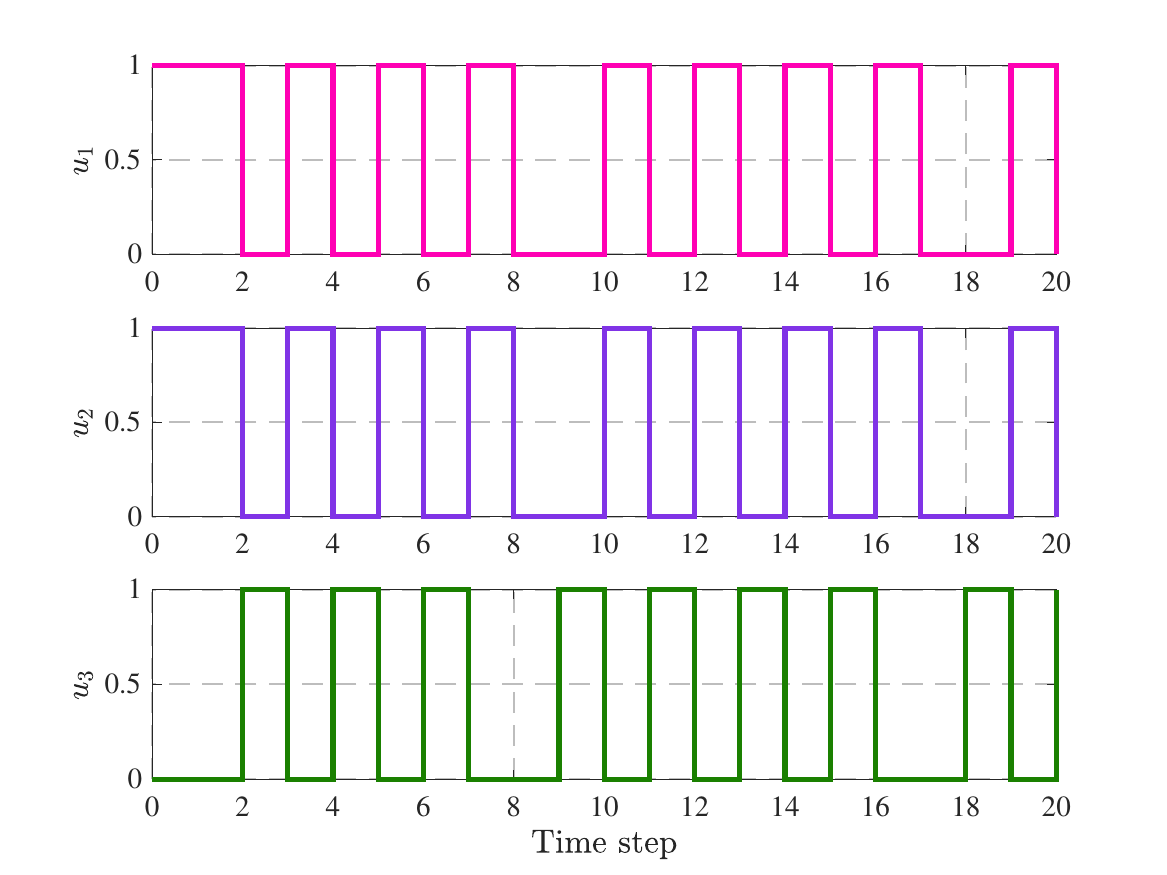} }}%
	\caption{Closed-loop state trajectories of three arbitrary rooms by designing  controllers via their data-driven symbolic models. While two out of the three initial conditions are positioned at the boundaries of the safe set, the designed controllers perfectly maintain the trajectories within the safe set.}%
	\label{fig:room}%
\end{figure}

\subsection{Vehicle network} 
As the second case study, we illustrate the efficacy of our data-driven findings for a vehicle network comprising $> 5000$ vehicles, characterized by unknown mathematical models and an unknown interconnection topology~\cite{sadraddini2017provably}.
 
 The state evolution within the interconnected network is described by the following formulation:
 	\begin{align*}
 		\Upsilon: x(k+1) = A x(k) + u(k),
 	\end{align*}
 	where the matrix $A$ includes diagonal blocks, denoted by $\hat{A}$, and off-diagonal blocks, depending on the interconnection topology, represented as $A_{i(i-1)} = A_w$ or zero matrix for $i \in \{2, \ldots, M\}$, defined as follows:
 	\[
 	\hat{A} = \begin{bmatrix} 1 & -1 \\ 0 & 1 \end{bmatrix}\!\!, \quad A_w = \begin{bmatrix} 0 & \tau \\ 0 & 0 \end{bmatrix}\!\!,
 	\]
 	where $\tau = 0.005$ denotes the strength of interconnection. We also define the vectors $x(k) = \left[x_1(k); \ldots ; x_{M}(k)\right]$ and $u(k) = \left[u_1(k); \ldots ; u_{M}(k)\right]$.
 	For each $i \in \{1, \ldots, M\}$, the state evolution of an individual vehicle is given by:
 	\[
 	\Upsilon_i: x_i(k+1) = \hat{A} x_i(k) + u_i(k) + A_w w_i(k).
 	\]
    Consequently, the network structure $\Upsilon$ is represented as $\Upsilon = \mathscr{N}\!\left(\Upsilon_1, \ldots, \Upsilon_{M}\right)$.
 
 The state of each vehicle is defined as $x_i \coloneqq \left[d_i;v_i\right], i \in \{1, \dots, M\}$, where the distance between vehicle $i$ and its preceding vehicle $i - 1$ is represented by $d_i$. Our objective is to compositionally synthesize controllers $u_i = [u_{i_1};  u_{i_2}]$, with $u_{i_1}, u_{i_2} \in \{-1, -0.8, -0.6, \dots, 1\}$, that force the network's states to stay inside the safe region $X_i = [0,\, 1] \times [-0.15, \, 0.55]$.
 
 \begin{figure}[t!]
 	\centering
 	\subfloat[\centering First scenario \label{subfig2}]{{\includegraphics[width=0.77\linewidth]{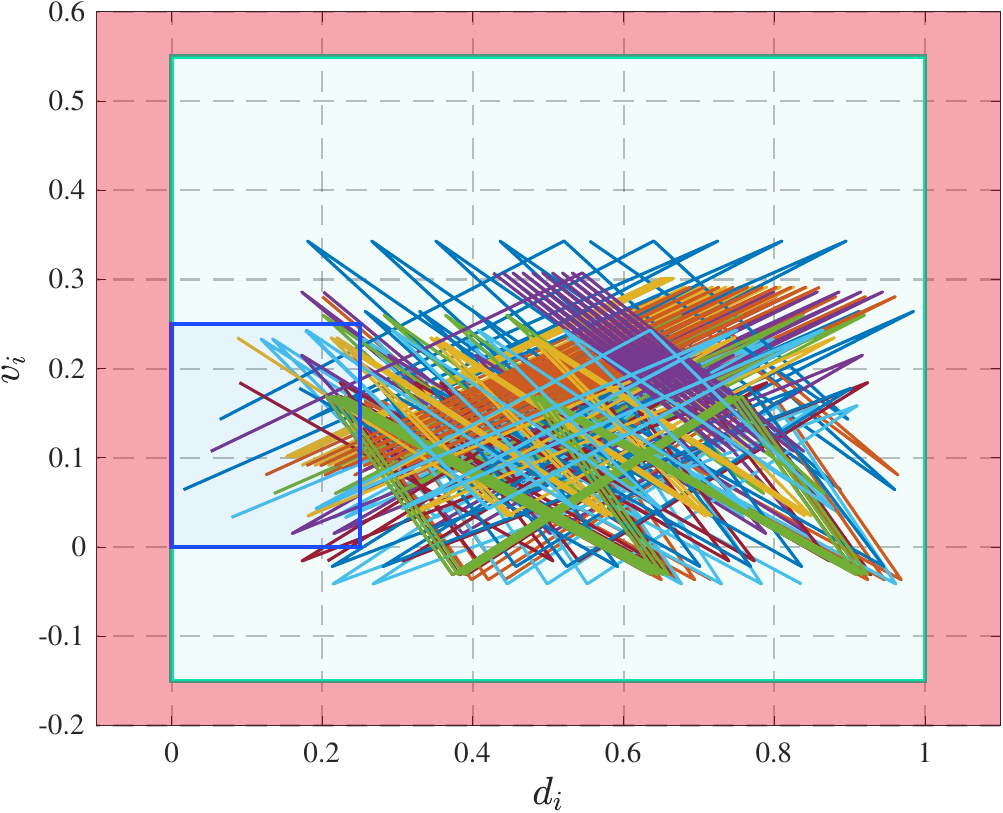} }}%
 	
 	\subfloat[\centering Second scenario \label{subfig3}]{{\includegraphics[width=0.8\linewidth]{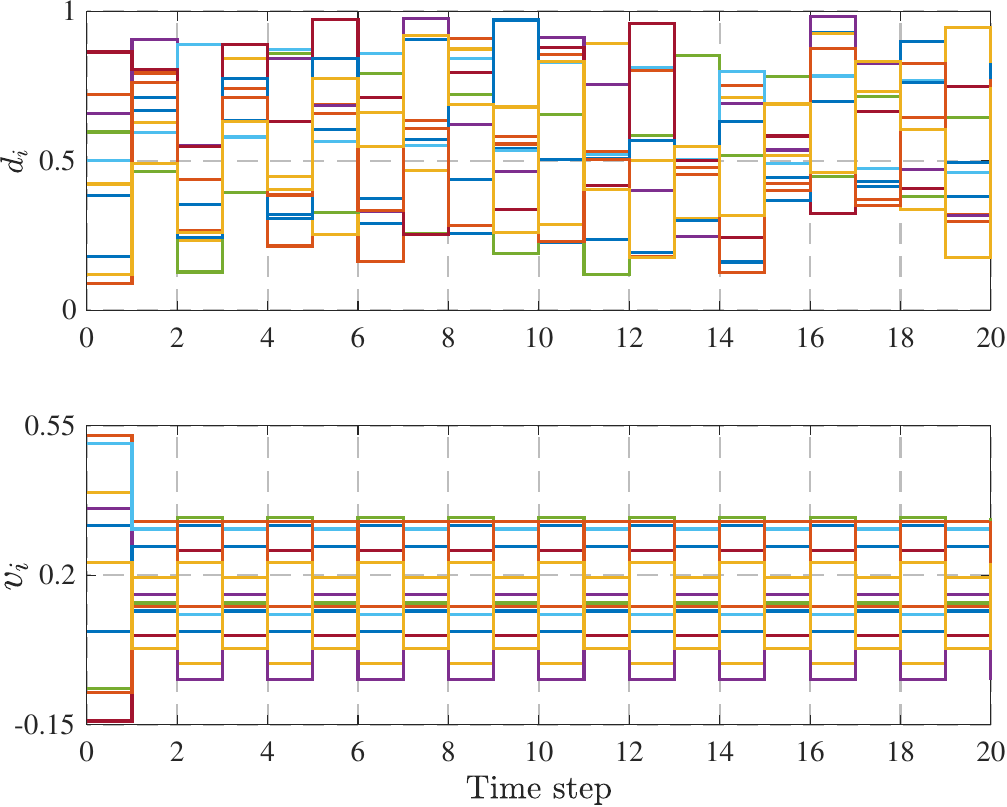} }}%
 	\caption{In the first scenario, we assume that each vehicle's trajectories start from a specific initial set \legendsquare{initt!25}. As shown in Figure~\ref{subfig2}, even though this initial set shares a border with the unsafe set \legendsquare{unsafe!35}, all trajectories remain within the safe set \legendsquare{state!35} and never enter the unsafe set. In the second scenario, we assume that trajectories of vehicles can start anywhere inside the safe set. As illustrated in Figure~\ref{subfig3}, all trajectories never leave the safe set under any circumstances. While only $10$ vehicles are selected for demonstration purposes, we observed that safety is maintained for all $M$ vehicles. \label{fig:platoon}}%
 \end{figure}
 We set the structure of our ASBF as $\mathcal{V}_i(q_i, x_i, \hat x_i) = \sum_{j = 1}^{5}q_i^j (d_i - \hat d_i)^{(5 - j)}+ \sum_{j = 6}^{10}q_i^j (v_i - \hat v_i)^{(10 - j)}$. After collecting data and  computing $\sigma_{i} = 0.3$, we solve SOP~\eqref{SOP} for all $i \in \{1, \dots, M\}$ and obtain
 \begin{align}\notag
 	\mathcal{V}_i^\ast(q_i,x_i,\hat{x}_i) = \: & 0.008(d_i - \hat d_i)^4 - 0.0203(d_i - \hat d_i)^3 \\&+ 0.0235(d_i - \hat d_i)^2 + 0.008(d_i - \hat d_i)\notag\\
 	& +0.3458(v_i - \hat v_i)^4 + 0.008(v_i - \hat v_i)^3\notag\\&+0.008(v_i - \hat v_i)^2-0.01(v_i - \hat v_i)+2,\label{new999}\\\notag
 	& \hspace{-2.3cm}  \mu_i^\ast = -0.7717, \quad \varpi_i^\ast = 10^{-6},
 \end{align}
 with a fixed $\gamma_{i}^\ast = 0.99$. We now calculate $\mathscr{L}_{i} = \max\{1.5753,0.5359\}=1.5753$ according to Algorithm \ref{Alg:1}.  Since $\mu_i^\ast + \varpi_i^\ast + \mathscr{L}_{i}\sigma_i = -0.2991 \leq 0$, according to Theorem~\ref{Thm1}, $\mathcal V^\ast(q, x, \hat x)= \sum_{i = 1}^{M}\mathcal V_i^\ast(q_i,x_i,\hat{x}_i)$, with $\mathcal V_i^*$ as in~\eqref{new999}, is an ABF between the vehicle network and its symbolic model.
 
 We now proceed with compositionally designing a controller for the vehicle network using data-driven symbolic models with the objective of keeping each vehicle's states within the predefined safe set $X_i = [0,\, 1] \times [-0.15, \, 0.55]$. Figure~\ref{fig:platoon} demonstrates that the synthesized controller can constrain the trajectories of representative vehicles within the designated safe region, whether they start from an initial set that shares a boundary with the unsafe set (\emph{i.e.,} Figure~\ref{subfig2}) or originate from arbitrary initial conditions within the safe set (\emph{i.e.,} Figure~\ref{subfig3}).
 
 \BS{\subsection{Highly nonlinear network}\label{subsec:HN_Case}}
  \BS{To better illustrate the advantage of our \emph{direct} data-driven technique, we now consider an complicated network, featuring \emph{non-affine, highly complex} dynamics and \emph{nonlinear} interconnection constraints, which render the use of system identification approaches particularly challenging. This network consists of more than $1000$ interconnected subsystems, denoted by $\Upsilon = \mathscr{N}(\Upsilon_1, \ldots, \Upsilon_{M})$. Both the interconnection structure and the mathematical models of the subsystems are assumed to be \emph{unknown}. The dynamics of each subsystem are given by
 \begin{align*}
 	\Upsilon_i : x_i(k + 1) = & \; 0.6 \tanh(x_i(k)) + 0.06 \sin(x_i(k) + u_i(k))\\& + 0.03 u_i^3(k) + w_i(k),
 \end{align*}
 where \[
 	w_i(k) \! = \! 0.16\big(x_{i-1}^3(k) + x_{i+1}^3(k)\big) + \ln\!\big(1 + (x_{i-1}(k)x_{i+1}(k))^2\big), 
 	\]
	with $x_0 \coloneq x_M$ and $x_{M+1} \coloneq x_1$.
 }

 \BS{The primary objective is to construct ASBFs and their corresponding symbolic models by solving SOP~\eqref{SOP}, and subsequently to compositionally derive a data-driven ABF based on the result of Theorem~\ref{Thm1}. Utilizing the obtained data-driven symbolic models, we then design controllers within the control input set $U_i = \{-2, -1.5, \dots, 1.5, 2\}$ that regulate each state $x_i$ within the predefined safe set $X_i = [-0.25, \, 0.25]$, while maintaining formal correctness guarantees.}
 
 \BS{To achieve this, we define the structure of the ASBF as 
 	\(
 	\mathcal{V}_i(q_i, x_i, \hat{x}_i) = q_{i}^1(x_i - \hat{x}_i)^6 + q_{i}^2(x_i - \hat{x}_i)^4 + q_{i}^3(x_i - \hat{x}_i)^2 + q_{i}^4.
 	\)
 	Following the procedure outlined in Algorithm~\ref{Alg:2}, we collect data and compute $\sigma_i = 0.05$. By solving SOP~\eqref{SOP} for each $i \in \{1, \dots, M\}$, we obtain $\mu_i^\ast = -0.0233$, $\varpi_i^\ast = 10^{-6},$ and
 	\begin{align}
 		&\mathcal{V}_i^\ast(q_i,x_i,\hat{x}_i) =  0.5(x_i - \hat{x}_i)^6 + 1.5(x_i - \hat{x}_i)^4\notag \\ &\hspace{2.2cm}- 0.0564(x_i - \hat{x}_i)^2 + 0.1796,\label{eq:R3_V}
 	\end{align}
 	with a set $\gamma_i^\ast = 0.91$. Using Algorithm~\ref{Alg:1}, we compute $\mathscr{L}_i = \max\{0.1754, 0.3133\} = 0.3133$. Since $\mu_i^\ast + \varpi_i^\ast + \mathscr{L}_i\sigma_i = -0.0076 \leq 0$ for all $i \in \{1, \dots, M\}$, it follows from Theorem~\ref{Thm1} that 
 	\(
 	\mathcal{V}(q, x, \hat{x}) = \sum_{i=1}^{M}\mathcal{V}_i^\ast(q_i, x_i, \hat{x}_i),
 	\)
 	with $\mathcal{V}_i^\ast$ as in~\eqref{eq:R3_V}, serves as an ABF between the unknown network and its symbolic model.}
 
 \BS{Using the constructed data-driven symbolic models, we then compositionally design a controller for the interconnected network such that each subsystem's state $x_i$ remains within the safe set $X_i = [-0.25, \, 0.25]$. Specifically, we first synthesize a local controller for each subsystem via its symbolic model using \texttt{SCOTS}~\cite{rungger2016scots}, and subsequently refine these controllers for the original unknown subsystems using the data-driven ASBFs. The overall network controller is a vector whose components correspond to the controllers of individual subsystems. The state trajectories of four representative subsystems and their respective control inputs are shown in Figure~\ref{fig:new_case}.}
 
  \begin{figure}[t!]
 	\centering
 	\subfloat[\centering \BS{Evolution of each subsystem} \label{fig1}]{{\includegraphics[width=0.7\linewidth]{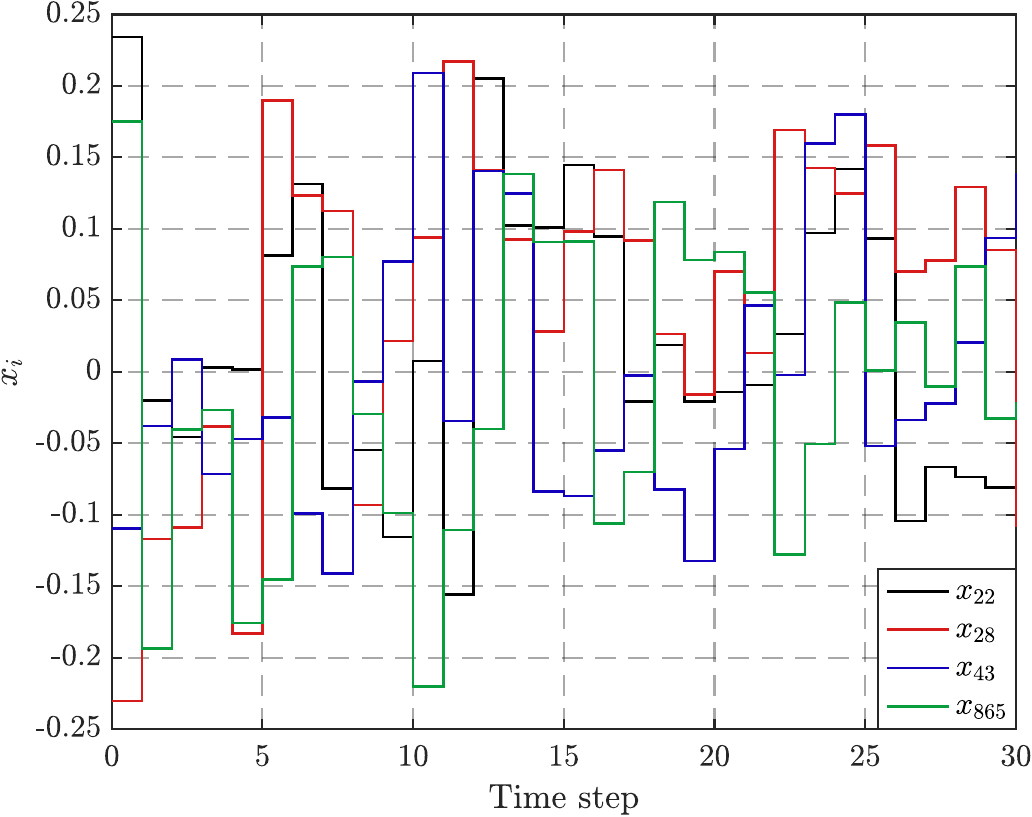} }}%
 	
 	\subfloat[\centering \BS{Control input of each subsystem} \label{fig2}]{{\includegraphics[width=0.7\linewidth]{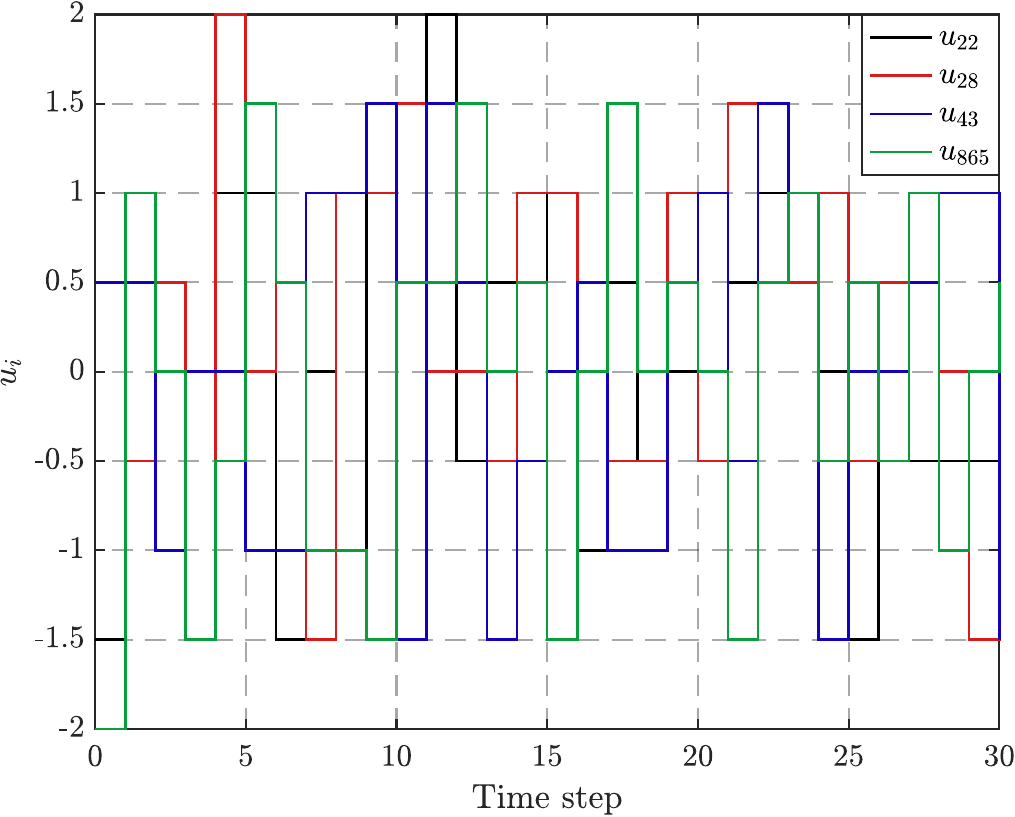} }}%
 	\caption{\BS{State trajectories and corresponding control inputs of four representative subsystems, obtained through controllers synthesized via data-driven symbolic models.} \label{fig:new_case}}%
 \end{figure}
\vspace{0.3cm}
\subsection{Sample Complexity Analysis}
We conduct a thorough examination of sample complexity within the context of monolithic (\emph{i.e.,} addressing the problem directly across the network in one shot) and compositional frameworks. To achieve this, we visually depict the relation between the required amount of data and the number of subsystems for both methodologies in Figure~\ref{Fig2}. It is clear that our data-driven divide-and-conquer approach significantly reduces sample complexity by aligning it with the granularity of subsystems. Consequently, as the number of subsystems increases, the growth in required data remains \emph{linear}. Conversely, in the monolithic approach, sample complexity increases \emph{exponentially} alongside the network's dimension, rendering it operationally impractical.

\begin{figure}[t!]
	\centering
	\resizebox{.8\linewidth}{!}{\begin{tikzpicture}
			\begin{loglogaxis}[
				xlabel=Number of subsystems,
				ylabel=Number of data,
				xmin = 10^0, xmax = 10^2.02,
				ymax = 10^300, ymin = 10^0,
				ytick={10^3,10^50,10^100,10^200,10^300},
				grid=both, major grid style={black!50}, legend style={at={(0.5,-2.75em)},anchor=north, draw=none, anchor=north, legend columns=-1, legend  style={/tikz/every even column/.append style={column sep=0.5cm}}}
				]
				\addplot[color = deeppink, mark = *, mark size = 1.25] coordinates {
					(1, 10^3) (2, 2 * 10^3) (3, 3 * 10^3) (4, 4 * 10^3) (5, 5 * 10^3) (6, 6 * 10^3) (7, 7 * 10^3) (8, 8 * 10^3) (9, 9 * 10^3) (10, 10 * 10^3) (11, 11 * 10^3) (12, 12 * 10^3) (13, 13 * 10^3) (14, 14 * 10^3) (15, 15 * 10^3) (16, 16 * 10^3) (17, 17 * 10^3) (18, 18 * 10^3) (19, 19 * 10^3) (20, 20 * 10^3) (21, 21 * 10^3) (22, 22 * 10^3) (23, 23 * 10^3) (24, 24 * 10^3) (25, 25 * 10^3) (26, 26 * 10^3) (27, 27 * 10^3) (28, 28 * 10^3) (29, 29 * 10^3) (30, 30 * 10^3) (31, 31 * 10^3) (32, 32 * 10^3) (33, 33 * 10^3) (34, 34 * 10^3) (35, 35 * 10^3) (36, 36 * 10^3) (37, 37 * 10^3) (38, 38 * 10^3) (39, 39 * 10^3) (40, 40 * 10^3) (41, 41 * 10^3) (42, 42 * 10^3) (43, 43 * 10^3) (44, 44 * 10^3) (45, 45 * 10^3) (46, 46 * 10^3) (47, 47 * 10^3) (48, 48 * 10^3) (49, 49 * 10^3) (50, 50 * 10^3) (51, 51 * 10^3) (52, 52 * 10^3) (53, 53 * 10^3) (54, 54 * 10^3) (55, 55 * 10^3) (56, 56 * 10^3) (57, 57 * 10^3) (58, 58 * 10^3) (59, 59 * 10^3) (60, 60 * 10^3) (61, 61 * 10^3) (62, 62 * 10^3) (63, 63 * 10^3) (64, 64 * 10^3) (65, 65 * 10^3) (66, 66 * 10^3) (67, 67 * 10^3) (68, 68 * 10^3) (69, 69 * 10^3) (70, 70 * 10^3)(71, 71 * 10^3) (72, 72 * 10^3) (73, 73 * 10^3) (74, 74 * 10^3) (75, 75 * 10^3) (76, 76 * 10^3) (77, 77 * 10^3) (78, 78 * 10^3) (79, 79 * 10^3) (80, 80 * 10^3) (81, 81 * 10^3) (82, 82 * 10^3) (83, 83 * 10^3) (84, 84 * 10^3) (85, 85 * 10^3) (86, 86 * 10^3) (87, 87 * 10^3) (88, 88 * 10^3) (89, 89 * 10^3) (90, 90 * 10^3) (91, 91 * 10^3) (92, 92 * 10^3) (93, 93 * 10^3) (94, 94 * 10^3) (95, 95 * 10^3) (96, 96 * 10^3) (97, 97 * 10^3) (98, 98 * 10^3) (99, 99 * 10^3) (100, 100 * 10^3)
				};
				\addplot[color = denim, mark = square*, mark size = 1.25] coordinates{ 
					(1, 10^3) (2, 10^6) (3, 10^9) (4, 10^12) (5, 10^15) (6, 10^18) (7, 10^21) (8, 10^24) (9, 10^27) (10, 10^30) (11, 10^33) (12, 10^36) (13, 10^39) (14, 10^42) (15, 10^45) (16, 10^48) (17, 10^51) (18, 10^54) (19, 10^57) (20, 10^60) (21, 10^63) (22, 10^66) (23, 10^69) (24, 10^72) (25, 10^75) (26, 10^78) (27, 10^81) (28, 10^84) (29, 10^87) (30, 10^90) (31, 10^93) (32, 10^96) (33, 10^99) (34, 10^102) (35, 10^105) (36, 10^108) (37, 10^111) (38, 10^114) (39, 10^117) (40, 10^120) (41, 10^123) (42, 10^126) (43, 10^129) (44, 10^132) (45, 10^135) (46, 10^138) (47, 10^141) (48, 10^144) (49, 10^147) (50, 10^150) (51, 10^153) (52, 10^156) (53, 10^159) (54, 10^162) (55, 10^165) (56, 10^168) (57, 10^171) (58, 10^174) (59, 10^177) (60, 10^180) (61, 10^183) (62, 10^186) (63, 10^189) (64, 10^192) (65, 10^195) (66, 10^198) (67, 10^201) (68, 10^204) (69, 10^207) (70, 10^210) (71, 10^213) (72, 10^216) (73, 10^219) (74, 10^222) (75, 10^225) (76, 10^228) (77, 10^231) (78, 10^234) (79, 10^237) (80, 10^240) (81, 10^243) (82, 10^246) (83, 10^249) (84, 10^252) (85, 10^255) (86, 10^258) (87, 10^261) (88, 10^264) (89, 10^267) (90, 10^270) (91, 10^273) (92, 10^276) (93, 10^279) (94, 10^282) (95, 10^285) (96, 10^288) (97, 10^291) (98, 10^294) (99, 10^297) (100, 10^300)
				};
				\legend{Compositional approach, Monolithic approach}
			\end{loglogaxis}
	\end{tikzpicture}}
	\caption{Sample complexity analysis. Plot is in logarithmic scale.} 
	\label{Fig2}
\end{figure}

\section{Conclusion}\label{concc}
In this paper, we developed a data-driven divide-and-conquer strategy for analyzing large-scale interconnected networks, characterized by both unknown mathematical models and interconnection topologies. Our approach treated an unknown network as a collection of individual subsystems and aimed to compositionally construct a symbolic model for the network by gathering data from the trajectories of its subsystems. The primary objective was to synthesize control strategies ensuring desired behaviors across unknown networks by utilizing local controllers, derived from data-driven symbolic models of individual agents. To achieve this, we employed alternating sub-bisimulation functions to quantify the similarity between state trajectories of each unknown agent and its data-driven symbolic model. Under the newly developed data-driven compositional conditions, we established an alternating bisimulation function between the unknown network and its symbolic model based on alternating sub-bisimulation functions of agents while ensuring correctness guarantees. Additionally, we illustrated that our data-driven compositional condition eliminates the need for the traditional small-gain condition, which typically requires precise knowledge of the interconnection topology.

\bibliographystyle{IEEEtran}
\bibliography{biblio}	

% Generated by IEEEtran.bst, version: 1.14 (2015/08/26)
\begin{thebibliography}{10}
\providecommand{\url}[1]{#1}
\csname url@samestyle\endcsname
\providecommand{\newblock}{\relax}
\providecommand{\bibinfo}[2]{#2}
\providecommand{\BIBentrySTDinterwordspacing}{\spaceskip=0pt\relax}
\providecommand{\BIBentryALTinterwordstretchfactor}{4}
\providecommand{\BIBentryALTinterwordspacing}{\spaceskip=\fontdimen2\font plus
\BIBentryALTinterwordstretchfactor\fontdimen3\font minus
  \fontdimen4\font\relax}
\providecommand{\BIBforeignlanguage}[2]{{%
\expandafter\ifx\csname l@#1\endcsname\relax
\typeout{** WARNING: IEEEtran.bst: No hyphenation pattern has been}%
\typeout{** loaded for the language `#1'. Using the pattern for}%
\typeout{** the default language instead.}%
\else
\language=\csname l@#1\endcsname
\fi
#2}}
\providecommand{\BIBdecl}{\relax}
\BIBdecl

\bibitem{bamieh2002distributed}
B.~Bamieh, F.~Paganini, and M.~A. Dahleh, ``Distributed control of spatially
  invariant systems,'' \emph{IEEE Transactions on Automatic Control}, vol.~47,
  no.~7, pp. 1091--1107, 2002.

\bibitem{jovanovic2005ill}
M.~R. Jovanovic and B.~Bamieh, ``On the ill-posedness of certain vehicular
  platoon control problems,'' \emph{IEEE Transactions on Automatic Control},
  vol.~50, no.~9, pp. 1307--1321, 2005.

\bibitem{tabuada2009verification}
P.~Tabuada, \emph{Verification and control of hybrid systems: {A} symbolic
  approach}.\hskip 1em plus 0.5em minus 0.4em\relax Springer Science \&
  Business Media, 2009.

\bibitem{zamani2011symbolic}
M.~Zamani, G.~Pola, M.~Mazo, and P.~Tabuada, ``Symbolic models for nonlinear
  control systems without stability assumptions,'' \emph{IEEE Transactions on
  Automatic Control}, vol.~57, no.~7, pp. 1804--1809, 2011.

\bibitem{coogan2016finite}
S.~Coogan, M.~Arcak, and C.~Belta, ``Finite state abstraction and formal
  methods for traffic flow networks,'' in \emph{Proceedings of IEEE American
  Control Conference (ACC)}, 2016, pp. 864--879.

\bibitem{majumdar2020abstraction}
R.~Majumdar, N.~Ozay, and A.-K. Schmuck, ``On abstraction-based controller
  design with output feedback,'' in \emph{Proceedings of the 23rd International
  Conference on Hybrid Systems: Computation and Control}, 2020, pp. 1--11.

\bibitem{swikir2019compositional}
A.~Swikir and M.~Zamani, ``Compositional synthesis of finite abstractions for
  networks of systems: A small-gain approach,'' \emph{Automatica}, vol. 107,
  pp. 551--561, 2019.

\bibitem{lavaei2019automated}
A.~Lavaei, ``Automated verification and control of large-scale stochastic
  cyber-physical systems: Compositional techniques,'' Ph.D. dissertation,
  Technische Universit{\"a}t M{\"u}nchen, 2019.

\bibitem{nejati2023formal}
A.~Nejati, ``Formal verification and control of stochastic hybrid systems:
  Model-based and data-driven techniques,'' Ph.D. dissertation, Technische
  Universit{\"a}t M{\"u}nchen, 2023.

\bibitem{hou2013model}
Z.-S. Hou and Z.~Wang, ``From model-based control to data-driven control:
  Survey, classification and perspective,'' \emph{Information Sciences}, vol.
  235, pp. 3--35, 2013.

\bibitem{dorfler2022bridging}
F.~D{\"o}rfler, J.~Coulson, and I.~Markovsky, ``Bridging direct and indirect
  data-driven control formulations via regularizations and relaxations,''
  \emph{IEEE Transactions on Automatic Control}, vol.~68, no.~2, pp. 883--897,
  2022.

\bibitem{kerschen2006past}
G.~Kerschen, K.~Worden, A.~F. Vakakis, and J.-C. Golinval, ``Past, present and
  future of nonlinear system identification in structural dynamics,''
  \emph{Mechanical Systems and Signal Processing}, vol.~20, no.~3, pp.
  505--592, 2006.

\bibitem{materassi2010topological}
D.~Materassi and G.~Innocenti, ``Topological identification in networks of
  dynamical systems,'' \emph{IEEE Transactions on Automatic Control}, vol.~55,
  no.~8, pp. 1860--1871, 2010.

\bibitem{nabi2012network}
M.~Nabi-Abdolyousefi and M.~Mesbahi, ``Network identification via node
  knockout,'' \emph{IEEE Transactions on Automatic Control}, vol.~57, no.~12,
  pp. 3214--3219, 2012.

\bibitem{shahrampour2014topology}
S.~Shahrampour and V.~M. Preciado, ``Topology identification of directed
  dynamical networks via power spectral analysis,'' \emph{IEEE Transactions on
  Automatic Control}, vol.~60, no.~8, pp. 2260--2265, 2014.

\bibitem{chaillet2014strong}
A.~Chaillet, D.~Angeli, and H.~Ito, ``Strong {iISS} is preserved under cascade
  interconnection,'' \emph{Automatica}, vol.~50, no.~9, pp. 2424--2427, 2014.

\bibitem{dashkovskiy2007iss}
S.~Dashkovskiy, B.~S. R{\"u}ffer, and F.~R. Wirth, ``An {ISS} small gain
  theorem for general networks,'' \emph{Mathematics of Control, Signals, and
  Systems}, vol.~19, pp. 93--122, 2007.

\bibitem{tanner2002stability}
H.~Tanner, V.~Kumar, and G.~Pappas, ``Stability properties of interconnected
  vehicles,'' in \emph{Proceedings of the 15th International Symposium on
  Mathematical Theory of Networks and Systems}, 2002.

\bibitem{mironchenko2020input}
A.~Mironchenko and C.~Prieur, ``Input-to-state stability of
  infinite-dimensional systems: {R}ecent results and open questions,''
  \emph{SIAM Review}, vol.~62, no.~3, pp. 529--614, 2020.

\bibitem{meyer2017compositional}
P.-J. Meyer, A.~Girard, and E.~Witrant, ``Compositional abstraction and safety
  synthesis using overlapping symbolic models,'' \emph{IEEE Transactions on
  Automatic Control}, vol.~63, no.~6, pp. 1835--1841, 2017.

\bibitem{lavaei2022dissipativity}
A.~Lavaei and M.~Zamani, ``From dissipativity theory to compositional synthesis
  of large-scale stochastic switched systems,'' \emph{IEEE Transactions on
  Automatic Control}, vol.~67, no.~9, pp. 4422--4437, 2022.

\bibitem{2016Murat}
M.~Arcak, C.~Meissen, and A.~Packard, \emph{Networks of dissipative systems},
  ser. SpringerBriefs in Electrical and Computer Engineering.\hskip 1em plus
  0.5em minus 0.4em\relax Springer, 2016.

\bibitem{willems1972dissipative}
J.~C. Willems, ``Dissipative dynamical systems part {I}: {G}eneral theory,''
  \emph{Archive for Rational Mechanics and Analysis}, vol.~45, no.~5, pp.
  321--351, 1972.

\bibitem{haddad2008nonlinear}
W.~M. Haddad and V.~Chellaboina, \emph{Nonlinear dynamical systems and control:
  {A} {L}yapunov-based approach}.\hskip 1em plus 0.5em minus 0.4em\relax
  Princeton University Press, 2008.

\bibitem{hashimoto2022learning}
K.~Hashimoto, A.~Saoud, M.~Kishida, T.~Ushio, and D.~V. Dimarogonas,
  ``Learning-based symbolic abstractions for nonlinear control systems,''
  \emph{Automatica}, vol. 146, 2022.

\bibitem{makdesi2021efficient}
A.~Makdesi, A.~Girard, and L.~Fribourg, ``Efficient data-driven abstraction of
  monotone systems with disturbances,'' \emph{IFAC-PapersOnLine}, vol.~54,
  no.~5, pp. 49--54, 2021.

\bibitem{devonport2021symbolic}
A.~Devonport, A.~Saoud, and M.~Arcak, ``Symbolic abstractions from data: {A}
  {PAC} learning approach,'' in \emph{Proceedings of the 60th IEEE Conference
  on Decision and Control (CDC)}, 2021, pp. 599--604.

\bibitem{coppola2022data}
R.~Coppola, A.~Peruffo, and M.~Mazo~Jr, ``Data-driven abstractions for
  verification of deterministic systems,'' \emph{arXiv:2211.01793}, 2022.

\bibitem{coppola2024data}
------, ``Data-driven abstractions for control systems,''
  \emph{arXiv:2402.10668}, 2024.

\bibitem{banse2024data}
A.~Banse, L.~Romao, A.~Abate, and R.~M. Jungers, ``Data-driven memory-dependent
  abstractions of dynamical systems via a {C}antor-{K}antorovich metric,''
  \emph{arXiv:2405.08353}, 2024.

\bibitem{kazemi2024data}
M.~Kazemi, R.~Majumdar, M.~Salamati, S.~Soudjani, and B.~Wooding, ``Data-driven
  abstraction-based control synthesis,'' \emph{Nonlinear Analysis: Hybrid
  Systems}, vol.~52, 2024.

\bibitem{ajeleye2023data}
D.~Ajeleye, A.~Lavaei, and M.~Zamani, ``Data-driven controller synthesis via
  finite abstractions with formal guarantees,'' \emph{IEEE Control Systems
  Letters}, vol.~7, pp. 3453--3458, 2023.

\bibitem{ajeleye2024data}
D.~Ajeleye and M.~Zamani, ``Data-driven construction of finite abstractions for
  interconnected systems: {A} compositional approach,''
  \emph{arXiv:2408.08497}, 2024.

\bibitem{lavaei2023symbolic}
A.~Lavaei, ``Symbolic abstractions with guarantees: {A} data-driven
  divide-and-conquer strategy,'' in \emph{Proceedings of the 62nd IEEE
  Conference on Decision and Control (CDC)}, 2023, pp. 7994--7999.

\bibitem{pola2016symbolic}
G.~Pola, P.~Pepe, and M.~D. Di~Benedetto, ``Symbolic models for networks of
  control systems,'' \emph{IEEE Transactions on Automatic Control}, vol.~61,
  no.~11, pp. 3663--3668, 2016.

\bibitem{baier2008principles}
C.~Baier, J.~P. Katoen, and K.~G. Larsen, \emph{Principles of model
  checking}.\hskip 1em plus 0.5em minus 0.4em\relax MIT press, 2008.

\bibitem{pola2009symbolic}
G.~Pola and P.~Tabuada, ``Symbolic models for nonlinear control systems:
  {A}lternating approximate bisimulations,'' \emph{SIAM Journal on Control and
  Optimization}, vol.~48, no.~2, pp. 719--733, 2009.

\bibitem{angeli2002lyapunov}
D.~Angeli, ``A {L}yapunov approach to incremental stability properties,''
  \emph{IEEE Transactions on Automatic Control}, vol.~47, no.~3, pp. 410--421,
  2002.

\bibitem{tran2016incremental}
D.~N. Tran, B.~S. R{\"u}ffer, and C.~M. Kellett, ``Incremental stability
  properties for discrete-time systems,'' in \emph{Proceedings of the 55th IEEE
  Conference on Decision and Control (CDC)}, 2016, pp. 477--482.

\bibitem{girard2007approximation}
A.~Girard and G.~J. Pappas, ``Approximation metrics for discrete and continuous
  systems,'' \emph{IEEE Transactions on Automatic Control}, vol.~52, no.~5, pp.
  782--798, 2007.

\bibitem{rungger2016scots}
M.~Rungger and M.~Zamani, ``{SCOTS}: A tool for the synthesis of symbolic
  controllers,'' in \emph{Proceedings of the 19th ACM International Conference
  on Hybrid Systems: Computation and Control}, 2016, pp. 99--104.

\bibitem{wood1996estimation}
G.~Wood and B.~Zhang, ``Estimation of the {L}ipschitz constant of a function,''
  \emph{Journal of Global Optimization}, vol.~8, pp. 91--103, 1996.

\bibitem{cortes2008discontinuous}
J.~Cortes, ``Discontinuous dynamical systems,'' \emph{IEEE Control Systems
  Magazine}, vol.~28, no.~3, pp. 36--73, 2008.

\bibitem{willems2005note}
J.~C. Willems, P.~Rapisarda, I.~Markovsky, and B.~L.~M. De~Moor, ``A note on
  persistency of excitation,'' \emph{Systems \& Control Letters}, vol.~54,
  no.~4, pp. 325--329, 2005.

\bibitem{de2019formulas}
C.~De~Persis and P.~Tesi, ``Formulas for data-driven control: {S}tabilization,
  optimality, and robustness,'' \emph{IEEE Transactions on Automatic Control},
  vol.~65, no.~3, pp. 909--924, 2019.

\bibitem{martin2023guarantees}
T.~Martin, T.~B. Sch{\"o}n, and F.~Allg{\"o}wer, ``Guarantees for data-driven
  control of nonlinear systems using semidefinite programming: {A} survey,''
  \emph{Annual Reviews in Control}, vol.~56, 2023.

\bibitem{samari2024abstraction}
B.~Samari, M.~Zaker, and A.~Lavaei, ``Abstraction-based control of unknown
  continuous-space models with just two trajectories,'' \emph{7th Annual
  Learning for Dynamics \& Control Conference (L4DC 2025), arXiv:2412.03892},
  2025.

\bibitem{makdesi2023data}
A.~Makdesi, A.~Girard, and L.~Fribourg, ``Data-driven models of monotone
  systems,'' \emph{IEEE Transactions on Automatic Control}, vol.~69, no.~8, pp.
  5294--5309, 2023.

\bibitem{lavaei2022data}
A.~Lavaei and E.~Frazzoli, ``Data-driven synthesis of symbolic abstractions
  with guaranteed confidence,'' \emph{IEEE Control Systems Letters}, vol.~7,
  pp. 253--258, 2022.

\bibitem{jahanshahi2023data}
N.~Jahanshahi and M.~Zamani, ``Data-driven synthesis of safety controllers for
  partially-observable systems with unknown models,'' in \emph{Proceedings of
  the 62nd IEEE Conference on Decision and Control (CDC)}, 2023, pp.
  1052--1057.

\bibitem{hernandez2001chebyshev}
M.~A. Hern{\'a}ndez, ``{C}hebyshev's approximation algorithms and
  applications,'' \emph{Computers \& Mathematics with Applications}, vol.~41,
  no. 3-4, pp. 433--445, 2001.

\bibitem{sadraddini2017provably}
S.~Sadraddini, S.~Sivaranjani, V.~Gupta, and C.~Belta, ``Provably safe cruise
  control of vehicular platoons,'' \emph{IEEE Control Systems Letters}, vol.~1,
  no.~2, pp. 262--267, 2017.

\end{thebibliography}

\begin{IEEEbiography}[{\includegraphics[width=1in,height=1.3in,clip,keepaspectratio]{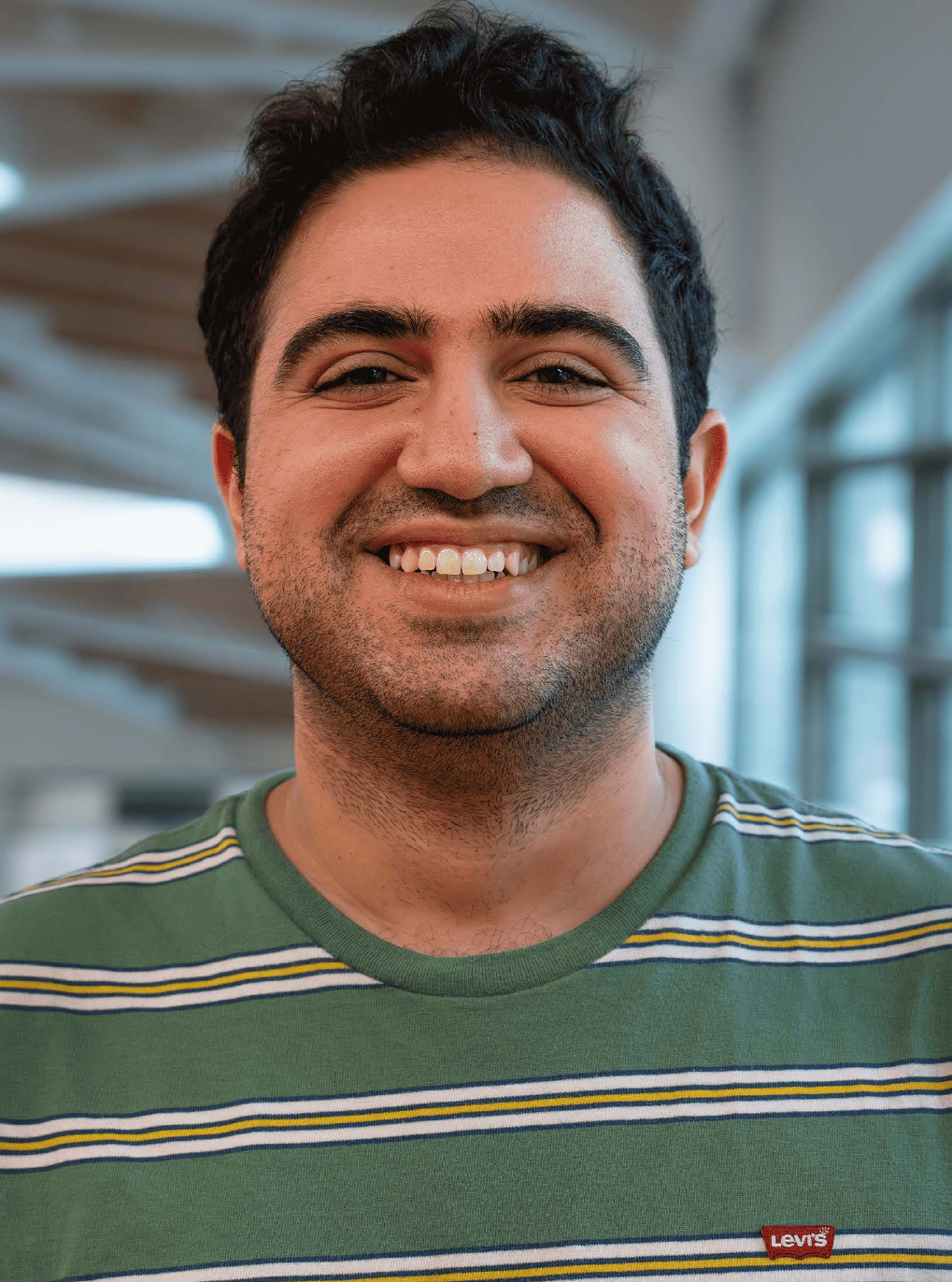}}]{Behrad Samari}~(Student Member, IEEE) received his B.Sc. and M.Sc. degrees in electrical engineering, control major, from K. N. Toosi University of Technology, Tehran, Iran, and University of Tehran (UT), Tehran, Iran, in 2019 and 2022, respectively. He is currently pursuing his PhD in the School of Computing at Newcastle University, U.K. He is the Best Repeatability Prize Finalist at the 8$^{\text{th}}$ IFAC Conference on Analysis and Design of Hybrid Systems (ADHS), 2024. His research interests include (nonlinear) control and system theory, data-driven approaches, and formal methods.
\end{IEEEbiography}

\begin{IEEEbiography}[{\includegraphics[width=1in,height=1.25in,clip,keepaspectratio]{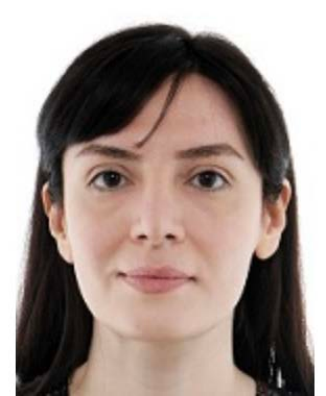}}]{Amy Nejati}~(M'18--SM'25) is an Assistant Professor in the School of Computing at Newcastle University in the United Kingdom. Prior to this, she was a Postdoctoral Associate at the Max Planck Institute for Software Systems in Germany from July 2023 to May 2024. She also served as a Senior Researcher in the Computer Science Department at the Ludwig Maximilian University of Munich (LMU) from November 2022 to June 2023. She received the PhD in Electrical Engineering from the Technical University of Munich (TUM) in 2023. She has received the B.Sc. and M.Sc. degrees both in Electrical Engineering. Her line of research mainly focuses on developing efficient (data-driven) techniques to design and control highly-reliable autonomous systems while providing mathematical guarantees. She was selected as a Best Repeatability Prize Finalist at ACM HSCC 2025 and as one of the CPS Rising Stars 2024.
\end{IEEEbiography}

\begin{IEEEbiography}
	[{\includegraphics[width=1in,height=1.25in,clip,keepaspectratio]{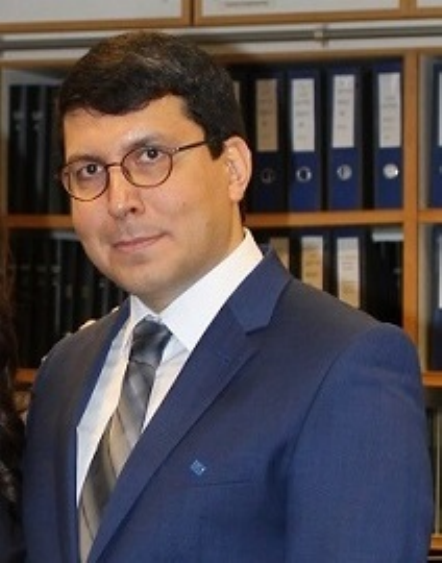}}]{Abolfazl Lavaei}~(M'17--SM'22)
	is an Assistant Professor in the School of Computing at Newcastle University, United Kingdom. Between January 2021 and July 2022, he was a Postdoctoral Associate in the Institute for Dynamic Systems and Control at ETH Zurich, Switzerland. He was also a Postdoctoral Researcher in the Department of Computer Science at LMU Munich, Germany, between November 2019 and January 2021. He received the Ph.D. degree in Electrical Engineering from the Technical University of Munich (TUM), Germany, in 2019. He obtained the M.Sc. degree in Aerospace Engineering with specialization in Flight Dynamics and Control from the University of Tehran (UT), Iran, in 2014. He is the recipient of several international awards in the acknowledgment of his work including ADHS Best Repeatability Prize (Finalist) 2024 and 2021, HSCC Best Demo/Poster Awards 2022 and 2020, IFAC Young Author Award Finalist 2019, and Best Graduate Student Award 2014 at University of Tehran with the full GPA (20/20). His research interests revolve around the intersection of Control Theory, Formal Methods, and Statistical Learning Theory.
\end{IEEEbiography}

\end{document}